\documentclass[
 reprint,
superscriptaddress,
frontmatterverbose,
preprintnumbers,
nofootinbib,
nobibnotes,
 amsmath,amssymb,
 aps,
pra,
prb,
rmp,
prstab,
prstper,
floatfix,
]{revtex4-2}
\usepackage{CJK}
\usepackage{amsmath}
\usepackage{graphicx}
\usepackage{srcltx}
\usepackage{amssymb,latexsym,amsfonts, amsmath}
\usepackage[english]{babel}
\usepackage{amsthm}
\usepackage{mathrsfs}
\usepackage{color}
\usepackage{dsfont}
\usepackage[normalem]{ulem}
\usepackage{verbatim}
\numberwithin{equation}{section}
\usepackage[english]{babel}
\usepackage{microtype}
\usepackage{amsmath}
\usepackage{tikz}
\usetikzlibrary{positioning,decorations.pathreplacing,shapes}
\usetikzlibrary{arrows,intersections}
\usepackage{verbatim}
\usepackage{tikz}
\usetikzlibrary{positioning,decorations.pathreplacing,shapes}
\usetikzlibrary{arrows,intersections}
\usepackage{subcaption}
\newcommand{\trace}{\mathrm{tr}}
\makeatletter

\newcommand{\ic}{\dot{\iota}}

\renewcommand{\@makefnmark}{}

\makeatother
\theoremstyle{definition}

\newtheorem{remark}{Remark}[section]

\newtheorem{theorem}{Theorem}[section]

\usepackage[user,xr]{zref}
\usepackage{xr-hyper}
\usepackage{hyperref}

\begin{document}
\title{The Transition from Quantum to Classical in weak measurements \\
and reconstruction of Quantum Correlation}

\author{Vadim V. Vorobyov}%
\email{v.vorobyov@pi3.uni-stuttgart.de}
\affiliation{3rd Institute of Physics, IQST and Centre for Applied Quantum Technologies, University of Stuttgart, 70569, Stuttgart, Germany}%

\author{Jonas Meinel}%
\affiliation{3rd Institute of Physics, IQST and Centre for Applied Quantum Technologies, University of Stuttgart, 70569, Stuttgart, Germany}%
\affiliation{Max-Planck Institute for Solid State Research, Stuttgart, Germany}%

\author{Hitoshi Sumiya}%
\affiliation{Advanced Materials Laboratory, Sumitomo Electric Industries Ltd., Itami 664-0016, Japan}%

\author{Shinobu Onoda}%
\affiliation{Takasaki Advanced Radiation Research Institute, National Institutes for Quantum and Radiological Science and Technology, Takasaki 370-1292, Japan}%

\author{Junichi Isoya}%
\affiliation{Faculty of Pure and Applied Sciences, University of Tsukuba, Tsukuba 305-8573, Japan}%

\author{Oleg Gulinsky}%
\email{gulinskyoleg8@gmail.com}
\noaffiliation

\author{Jörg Wrachtrup}%
\email{j.wrachtrup@pi3.uni-stuttgart.de}
\affiliation{3rd Institute of Physics, IQST and Centre for Applied Quantum Technologies, University of Stuttgart, 70569, Stuttgart, Germany}%
\affiliation{Max-Planck Institute for Solid State Research, Stuttgart, Germany}%

\underline{\textbf{}}
\begin{abstract}
The ability to track and control the dynamics of a quantum system is the key to quantum technology.
Despite its central role, the quantitative reconstruction of the dynamics of a single quantum system from the macroscopic data of the associated  observable remains a problem.  We consider this problem in the context of weak measurements of a single nuclear carbon spin in a diamond with an electron spin as a meter at room temperature,
which is a well-controlled and understandable bipartite quantum system  \cite{Hansen}, \cite{TaminiauT},\cite{Doherty}.
 In this work, based on a detailed theoretical analysis of the model of  experiment, we study the relationship between the statistical properties of the macroscopic readout signal of the spin of a single electron and the quantum dynamics of the spin of a single nucleus, which is characterized by a parameter associated with the strength of the measurement. We determine the parameter of measurement strength in separate experiments and use it to reconstruct the quantum correlation. 
   We control the validity of our approach applying the Leggett-Garg test.

\end{abstract}

\maketitle

\section{Introduction}

The advantage of quantum information systems is based on
the ability to store and process information encoded
in a set of qubits. An obstacle to the implementation of complex quantum algorithms is the decoherence effect that occurs in the process of interaction and transformation of information.
To transfer information from one subsystem of a complex system to another, both subsystems must be entangled, which inevitably leads to irreversible state changes and loss of information.
The ability to track and control the dynamics of a quantum system is the key to quantum technology.

To the best of our knowledge, it is impossible to completely and unambiguously
restore the state of quantum systems  based on the macroscopic measurements of  observables associated with the system.
  In principle, the state of a quantum system can be characterized using the Wigner distribution function.
However, the reconstruction procedure for this function is extremely sensitive to measurement errors, is inefficient, and requires a large number of experiments \cite{Smithey}, \cite{Lvovsky}.  In this work, we pose a more modest problem of restoring the correlation function of the dynamics of an observable associated with  a quantum object. We control the accuracy of our reconstruction by checking the non-classical properties of the resulting correlation function using the Leggett-Garg test.

In contrast to the classical filtering and control problem, the main difficulty in the quantum case is that the probabilistic properties of the classical (macroscopic) process at the output of the system and the quantum process that generates it differ significantly (see e.g.  \cite{EPR},  \cite{Bush}). In this regard, it is interesting to trace theoretically the  transformation of a quantum process into a macroscopic output process and test this analysis with experimental results on a simple model.
On this path, we are guided by the idea that measurements are not an instantaneous jump-like act (“collapse” of the wave function), but a process in which one quantum state is replaced by another, pure or mixed, under the control of some interaction Hamiltonian, while the final stage of this process is just a non-unitary nature (see section \ref{bhnghXY}).

The classical probability measure $\mathds{P}$ is defined on  the probability space $(\Omega,\mathscr{F})$ or, equivalently,  on a Boolean algebra $\mathscr{A}$ (distributive lattice) and satisfies the strong additivity property
\begin{equation}\label{StrogA}
 \mathds{P}(A)+\mathds{P}(B)=\mathds{P}(A\wedge B)+\mathds{P}(A\vee B), \quad A, B\in \mathscr{A}.
\end{equation}
 On the other hand, the quantum probability is defined
   on an orthomodular (non-distributive) lattice generated by projections of a Hilbert space,  where a finite function (dimension function) with a property similar to \ref{StrogA} exists only in special cases
  (see \cite{Birkhoff}, and   \cite{Redei}).
  Moreover,  there is no dispersion free state (corresponding to the
$\delta$-measures in classical case) on a complex Hilbert space of
dimension greater than 2 (Gleason's theorem \cite{Glsn}, see also \cite{Bunce}). In other words, randomness is an inevitable property of quantum states. 

Strong additivity is a distinctive property that fundamentally separates classical probability from quantum probability, which manifests itself in a variant of the proof of Wigner and d'Espagnat's formulation \cite{Wigner}, \cite{Espagnat} for Bell's theorem, given for clarity in Appendix \ref{MatBell}.

   J. Bell \cite{Bell1}, \cite{Bell2}, using Bohm's version \cite{Bohm1}(ch.22)  of the Einstein-Podolsky-Rosen (EPR) arguments \cite{EPR},
   introduced test inequalities relating correlations between measurements in separate parts of the classical composite system. While Bell's inequalities examine correlations of bounded random variables over space,  the more recent Leggett-Garg inequality \cite{Leggett} examines correlations over time.
The simplest Bell and Leggett-Garg type inequalities can be represented in the same form
$$
C_{21}+C_{32}-C_{31}\leq 1,
$$
where $C_{ij}$ are the two-points correlation functions with different pairs of space or time arguments.
This inequality limits the strength of spatial or temporal correlations that can arise in a classical framework and is expected to be violated by quantum mechanical unitary dynamics such as Rabi oscillations or Larmor precession.

From a mathematical point of view, Bell-type inequalities arose in classical probability more than a hundred years before Bell's discovery. The first appearance of such a result is associated (see \cite{Pitowsky}, \cite{Hess-Leg2016})  with the name of the creator of Boolean algebra \cite{Boole}. The final solution to this problem is formulated as  Kolmogorov's consistency theorem (see e.g. \cite{Doob}).
  The key mathematical condition of Bell's type theorems is that all (bounded) random variables given in space or in time are defined on the same probability space, which is especially significant  in the case of spatially separated subsystems in the Bell test.
In terms of applicability, a critical limitation of Bell-type theorems is that inequalities are derived under the assumption of non-invasive measurements.

We consider the tracking problem in the framework of  measurements of a single nuclear  spin in a diamond with an electron spin as a meter at room temperature \cite{Hansen},\cite{TaminiauT},\cite{Doherty}.
 This is a well-controlled and understandable bipartite quantum system. A feature of  a substitutional nitrogen nuclear spin is that it lies on the NV axis and therefore its hyperfine tensor  is rotationally symmetric and collinear with the NV axis.

The possibility of quantum non-demolition
(QND) measurement of single nitrogen nuclear spins ($^{14}N$,$^{15}N$)  to probe different charge
states of the NV center was demonstrated  in \cite{neumannS}.
 In this case, the system observable is the nuclear spin $I_z$, which undergoes the Rabi oscillations, while the  probe observable is the electron spin $S_z$.
This is a realizable QND, which  is close to  the conditions of  projective measurement \cite{QND-PM}.

Tracking the precession of a single nuclear $^{13}C$
 spin  using periodic weak measurements was demonstrated in \cite{RBL}, \cite{Cujia}.  The nuclear spin of carbon in diamond weakly interacts with the electron spin of a nearby nitrogen vacancy center, which acts as an optically readable measurement qubit.
The nuclear spin undergoes a free precession around
the $z$ axis with an angular velocity given by the Larmor frequency
$\omega$.  The precession is monitored  by probing the nuclear spin component $I_x$ by means of a conditional rotation via the effective interaction Hamiltonian $H\approx A_\|  S_zI_x$, which  couples $I_x$ with $S_z$ component of electronic spin.

The experimental setup of our work basically coincides with the  scheme of \cite{RBL}. However, the problem of restoring the correlation function of a quantum signal from the data of successive (weak) measurements and, in particular, the subsequent application of the Leggett-Garg  test required a significant development of the experiment and an increase in the measurement accuracy.

To the best of our knowledge  (see \cite{Emery}, \cite{Hess-Leg2016}), most of the experiments on the application of the Leggett-Garg test used explicit or implicit pre-processing of macroscopic experimental data (application of an empirical "measurement factor", filtering, post-selection).
For example,  Palacios-Laloy  et al.  \cite{Palacios}  wrote (we adhere to the original notations): "under macrorealistic assumptions,
the only effect of the bandwidth of the resonator would be to reduce the measured signal by its Lorentzian response function  $C(\omega)=1/\big[1+(2\omega/k)^2\big];$  we thus have to correct for this effect
by dividing the measured spectral density $\tilde{S}_z(\omega)$
 by $C(\omega)$. We then compute $K(\tau)$ by inverse
Fourier transform of $S_z(\omega)=\tilde{S}_z(\omega)/C(\omega)$."
Palacios-Laloy et al.  found that their
qubit violated a LG-tests, albeit with a single data point, and conclude that their
system does not admit a realistic, non-invasively-measurable description.

 In Waldherr   et al.  \cite{Neumann}   the data pre-processing procedure is explicitly described.
They reconstruct the conditional Rabi oscillation for
 a substitutional nitrogen  nuclear spin  using (almost) projective QND  measurements \cite{neumannS}.
The application of a narrowband, nuclear spin state-selective microwave  $\pi$
pulse flips the electron spin into the $|+1\rangle_e$ state or $|0\rangle_e$ and $|-1\rangle_e$ states
conditional on the state of the nuclear spin. However,  only if the
measurement outcome is $|+1\rangle_n$,  the Rabi oscillation is generated.
Since the fluorescence intensity differs significantly for the electron spin states $|+1\rangle_e$ (low) and $|0\rangle_e$ and $|-1\rangle_e$ (high), these target states can be distinguished in the histogram of many subsequent measurements using a maximum likelihood statistical procedure. Low
fluorescence level indicates that the  $\pi$ pulse was successful,
i.e., that the nuclear spin state is $|+1\rangle_n$.
Only in this  successful case a resonant radio-frequency pulse of certain length is applied
and a subsequent measurement is used for data analysis.

In our experiment (see Section \ref{exper}),
the successive weak measurements  generate a classical stochastic signal whose characteristics can be calculated using a positive operator-valued measure (POVM) measurement scheme.
   We theoretically show that the correlation of the classical output signal can be converted to the correlation of nuclear spin dynamics using a mapping dependent on the factor $\alpha$ characterizing the strength of the measurements.
A similar idea of applying inverse imaging to macroscopic data has been used in the field of quantum tomography
 (see e.g. Smithey et al. \cite{Smithey}).

We apply the theoretical mapping to the output classical process (or, equivalently, to the empirical correlation function) in order to recover the correlation function of the quantum object. Only then do we apply the L-G inequality to verify that the reconstructed process violates the classical properties of correlations. Thus, although the L-G test was not developed for invasive (albeit weak) measurements, our approach allowed us to apply it to verify the compliance of the model and the method of analysis with the experimental conditions.

Moreover, our approach makes it possible to introduce
analogue of the classical relative entropy of Kullback and
Leibler as a measure of the discrepancy of information that
occurs during the measurement process. This definition differs from the relative entropy commonly used in the non-commutative case.


\emph{The paper is organized as follows.}
A detailed description of the main experiment and  the results are given in Section \ref{exper}, while
    the theoretical analysis  is carried out in Section \ref{Analysis}.

The main tools of our analysis are the POVM measurements theory  and  the Campbell-Baker-Hausdorff (BCH) formula. We  introduce a modification of  closed form \ref{DCHGu} of the BCH formula (a simple proof is given in the Appendix \ref{BCH-A}) and show that it can be used in our  model. In view of the numerous references in our main text to various aspects of the POVM theory, we found it useful  to summarize for convenience the main facts of this theory in the Appendix \ref{POVMs}.


\section{Experiment and results}\label{exper}

\subsection{Experimental setup}

We use a single $\mathrm{^{13}C}$ nuclear spin in diamond as quantum system, probed by an NV center electron spin (see Fig. \ref{fig1}a).
The electron spin interacts weakly with the single $\mathrm{^{13}C}$ nuclear spin  through their hyperfine interaction.
Each NV
electron spin can selectively address nuclear spins in the near
vicinity under control of the dynamical decoupling (DD) sequence, and  is then read out using projective optical measurements (see SM, section I for details).


\subsection{ Model and main assumptions}

 Recall that the secular hyperfine vector $\mathbf{A}_z$ of the hyperfine tensor $\mathbf{A}(\mathbf{r})$, $\mathbf{r}=(r, \theta,\phi)$ with a properly chosen $x$ axis is given by
\begin{equation}\label{TensorHF}
  \mathbf{A}_z = \Big(A_{xz},  0, A_{zz}\Big)=\big(A_\bot, 0, A_\|\big).
\end{equation}
 Since the electronic spin
precesses at a much higher frequency than the nuclear spin, the
nuclear only feels the static component of the electronic field. Therefore, the dynamics of the unified system with single spin $^{13}C$ can be  described by the Hamiltonian in the secular approximation:
\begin{equation}\label{MainHam}
  H=\omega I_z + 2 \pi A_\|S_z\otimes I_z + 2\pi A_\bot S_z\otimes  I_x,
\end{equation}
where $\omega_L =\gamma_C B_z$, $\gamma_C$ is the $^{13}C$ gyromagnetic ratio, $B_z$ is the external magnetic field along the NV axis, $S_z, I_x,I_z$ are the electron and  nuclear spin-1/2 operators, $A_\|= A_{zz}$ and $A_\bot=\sqrt{A_{xz}^2+A_{yz}^2}$ are the parallel and transverse hyperfine coupling parameters. After the interaction controlled by the Hamiltonian and the DD sequence, the amplitude of the nuclear component $I_{x}$
is mapped to
optically detectable component $S_z$  of the NV center.

In our experiment, the interaction between the electron and nuclear spins is controlled by a Knill-dynamical decoupling (KDD-XYn) sequence of periodically spaced $\pi$ pulses, which contains $n$ units of $20$ pulses with a pulse interval $\tau$.
The measurement strength of weak measurements can be varied.
If $\tau$ is adjusted to the effective nuclear Larmor period $\omega_L$, that is,
(see \cite{TaminiauT})
\begin{equation}\label{tau1}
  \tau=\frac{(2k+1)\pi}{2\omega_L + 2\pi A_{zz}}, \qquad \text{if}\quad 2\pi A_{zx}\ll\gamma_n B
\end{equation}
or
\begin{equation}\label{tau2}
  \tau=\pi/\omega_L\qquad \text{if}\quad 2\pi A_{zz}\ll\omega_L,
\end{equation}
 we can assume that the system evolves under the effective Hamiltonian \cite{Ma}, \cite{BossDegen}, \cite{RBL}
   \begin{equation}\label{HAeff}
  H_{\mathrm{eff}}  = 2 \alpha S_z \otimes I_x,
\end{equation}
where the Larmor frequency  is determined by the static magnetic field (in our experiment $B \approx 0.25 \, \mathrm{T}$).    The effect of the KDD-XYn sequence is specified by a measurement strength parameter $\alpha =  \pi N_p  A_\perp \tau$, which depends on the transverse hyperfine   component $A_\perp$  and  the number of pulses $N_p =20 n$.

Nuclear-spin precession at the Larmor frequency around an external magnetic field $\mathbf{B}$ is
perturbed by the presence of the electron spin. As a result,   depending on the NV charge and spin states $NV^-$,  $m_S =\{-1,0, +1\}$ or $NV^0$,  $m_S = \{-1/2, +1/2\}$, the precession frequency and axis are modified.
The interaction Hamiltonian
 gives rise to electron-spin-dependent nuclear precession frequencies
$\omega_0$ (electron spin in $|0\rangle$) and $\omega_{\pm1} =
\sqrt{(\omega_0 \pm a_\|)^2 + a_\bot^2}$
(electron in $|m_S = \pm1\rangle$, $a_\|= 2 \pi A_\| $, $a_\bot = 2 \pi A_\bot $ ).
 The interaction controlled by the effective Hamiltonian $H_{\mathrm{eff}}$ carries the main information about the nuclear spin  through the transverse hyperfine component $A_\perp$, but introduces a dephasing  associated with the backaction on the nuclear spin.
 At the same time,
the secular \emph{longitudinal part} of the hyperfine interaction $A_\|S_z\otimes I_z$ turns out to be a source of significant information distortion.

\subsection{Influence of the longitudinal part. Diamond sample}

One of the possible mechanisms for involving the  longitudinal component $A_\|$ into the dynamics in our experiment is the process of optical readout \cite{Cujia}.  During optical readout
 of the NV centre with \emph{non-resonant} laser pulse, due to the shift $\bigtriangleup\omega = |\omega_0 - \omega_{\pm1}|$ of nuclear precession, the NV
centre cycles through its electronic states before reaching the $m_S = 0$ spin-polarized
steady state.
The way to overcome the influence of the surrounding nuclear  spins $^{13}C$ on the sample spin
   is to reduce the concentration of nuclear spins $^{13}C$.
To this end, our experiment is carried out at NV centers in \emph{isotopically purified diamond} (see details in SM section II).

Uncontrolled dynamics of electron spin states during and after the readout leads to random $z$-rotations of the nuclear spin under the action of the secular \emph{longitudinal part} of the hyperfine interaction $A_\|S_z\otimes I_z$.
Therefore,  due to a random phase accumulation, stochastic flips of the sensor spin  lead to the decoherence of target spin with the intrinsic transverse relaxation time
  $T^*_2$  or the intrinsic  nuclear dephasing rate
  $
  \Gamma_{intr}=(T^*_{2})^{-1}
$
  It is also believed that the longitudinal  hyperfine interaction during laser readout causes  an additional dephasing with rate \cite{Cujia}
  \begin{equation}\label{Readoh}
    \Gamma_{opt}\approx\frac{A^2_\|t_l^2}{2t_s},
 \end{equation}
where $t_l=t_{readout}$ is a period at which $m_s\neq0$ and $t_s$    is a sample time.

The influence of the \emph{longitudinal part} of the hyperfine interaction $A_\|S_z\otimes I_z$ associated with optical readout accumulates with an increase in the number of measurements. As a result, a short period of weakly perturbed behavior does not make it possible to reveal the purely quantum properties of the nuclear spin.
Therefore, for the purposes of our experiment, it is desirable to choose a pair of NV-nuclear spins for which the value of $A_\|$ is negligibly small.
The properties of the various electron-nuclear pairs and the process of choosing a suitable pair are described in more detail in SM.
In our diamond sample, we managed to find such a pair corresponding to "the magic angle cone", which made it possible to obtain a clear picture of the interaction.

\subsection{Detection protocol}

We now turn to the description of the results of the experiment with this particular magic pair, designated as NV2.
In our experiment, we investigated both prepolarized and nonpolarized initial conditions for the nuclear spin.
  While electron spin polarization is a relatively easy task, nuclear spin polarization is  a rather delicate procedure.
In the main text, we analyze and present the main results for the partially polarized case, when incomplete polarization occurs during the first measurements (see section \ref{initST}). The results for the case of a prepolarized target spin are presented in the SM for additional information.

When initially the nuclear spin is not polarized, it is represented by a completely mixed state
\begin{equation}
\label{eq1}
\rho_I = \frac{1}{2}|0\rangle \langle 0|+ \frac{1}{2}|1\rangle \langle 1| =\frac{1}{2} I^\alpha+\frac{1}{2}
I^\beta = \frac{1}{2} \mathds{1}= I_e,
\end{equation}
where $I^\alpha  = I_e  + I_z$, $I^\beta  = I_e - I_z$ and  $I_e = \frac{1}{2}\mathds{1}$,  $S_k = I_k  = 1/2 \, \sigma_k$, and $\sigma_k$ are Pauli matrices. In section \ref{initST}, we show that the  initial state \ref{eq1} during the first measurements is transformed into a partially polarized state $I_e  + \sin \alpha  I_x$, the degree of polarization of which depends on the parameter $\alpha$.

The measurement protocol is as follows.
The electron spin is optically pumped (see Fig. \ref{fig1}a,b) into the state $|0\rangle=|m_S  = 0\rangle$ or equivalently to $S_e  + S_z$, and then rotated by a $(\pi/2)_y$ pulse around the y-axis to state $S_e  + S_x$.
Then, the interaction controlled by the KDD-XYn sequence is applied, which ends with a $(\pi/2)_x$ pulse along the $x$-axis.
In between two successive measurements  the target spin undergoes free precession around the $z$-axis with the Larmor frequency $\omega_L$ (see Fig. \ref{fig1}c), during which the accumulation of information about the target object takes place.
 The final optical readout of the $S_z$ component of the NV sensor repolarizes the sensor back to the initial state $S_e + S_z$ while maintaining the nuclear spin state
in the $x-y$ plane, which reduces the amplitude of
the $I_y$ component by a factor of $\cos\alpha$.
Nuclear precession
 leads to a mixing of the $I_x$ and $I_y$ amplitudes.
Thus, in the process of measurement, the components $I_x$ and $I_y$ are subject to a backaction, which leads to an exponential decay of the spin amplitude, depending on the strength parameter $\alpha$.

\subsection{Data analysis}

The experiment is designed in such a way that the interaction Hamiltonian (the effective Hamiltonian proportional to $S_z\otimes I_x$) and density matrices are expressed in terms of  operators $S_i$ and $I_j$ of the basis \ref{Basil} so  the dynamics of the object can be calculated using a modification of the   Baker-Campbell-Hausdorff (BCH) formula \ref{DCHGu}.

An analysis of the transformation of the state of the composite system and, in particular, the dynamics of the state of the target nuclear spin during the interaction is given in Section \ref{bhnghXY}.
We single out the moment  when, as a result of the interaction of the NV-center as a meter with the spin  $^{13}C,$ the amplitude of the observable $I_x$ is mapped into the observable $S_y$ and next to $S_z$ after $\pi/2$ pulse along $S_x$ axis with the inevitable uncertainty introduced by the factor $\sin \alpha,$ which characterizes the strength of the interaction.
Together with the projective measurement, this process converts the quantum signal into the classical stochastic process of the measured signal.

In section \ref{bhnghXY} a recurrent equation for the dynamics of a composite system \ref{Lindblad} and  the recurrent formula \ref{vfnhbwfU} for  the state of the target nuclear spin are obtained.  Under natural assumptions, we get approximate formulas specifying the values of the amplitudes of the observables $I_x $ and $I_y$ corresponding to the nuclear spin at an arbitrary instant of time.
The expression for the amplitude $x_N$ of the $I_x$ component, taking into account incomplete polarization with an indefinite sign, can be represented as (see \ref{Dinamik})
 \begin{equation}\label{DinamikI}
 \begin{aligned}
    x_N^\pm\approx \pm\sin\alpha\cos (\omega Nt_f)\exp\Big\{-\frac{(N-1)\alpha^2}{4}\Big\}\\
   \approx \pm\sin \alpha \cos(\omega N t_f)e^{-N\Gamma_\bot t_s},
     \end{aligned}
  \end{equation}
where $t_s$ is the total measurement sequence time (see Fig. \ref{fig1}a) and $\Gamma_\bot=\alpha^2/4 t_s$ is
the measurement-induced dephasing rate.
This formula is then used (see section \ref{theory}) to theoretically calculate the correlation functions (see \ref{Czr})
\begin{equation}\label{eq4I}
  C^{I_x}(N)=\cos(\omega N t_f)\exp\Big\{-\frac{(N-1)\alpha^2}{4}\Big\},
\end{equation}
for the quantum process corresponding to the component $I_x$ and the classical output process of the component $S_z$
\begin{equation}\label{CZ}
\begin{aligned}
  C^{S_z}(N)=\sin^2\alpha\cos (\omega Nt_f)\exp\Big\{-\frac{(N-1)\alpha^2}{4}\Big\}\\=\sin^2\alpha \  C^{I_x}(N).
  \end{aligned}
\end{equation}
First, we note that formulas \ref{eq4I} and  \ref{CZ} differ by the  factor $\sin^2 \alpha$ in which one $\sin\alpha$  is generated by incomplete polarization (see section \ref{initST}), and the second $\sin\alpha$ is given by the measurement process itself, while  both are determined by the strength parameter $\alpha$ of the weak measurements.
Further, although the correlation function $ C^{I_x}(N)$ corresponds to the quantum process of the $I_x$-component and its formula includes a pure component $\cos (\omega Nt_f)$ corresponding to the Larmor precession, this pure component is distorted by the exponential decay factor.

The above analysis shows that the application of the L-G test to the correlation functions of the output process does not make sense, since this test is designed for non-invasive measurements of the observable in the pure state.
However, we can use the  theoretical analysis to reconstruct the correlation function of the Larmor precession from experimental data. To do this, we need to be able to restore the  parameter $\alpha$ with great accuracy. In the following subsections \ref{tst}, \ref{RcI} and in the SM, we describe the procedures for estimating the  parameter $\alpha$ in various experimental modes.

Thus, we come  to a simple algorithm for analyzing the experimental data: 1) calculate  the empirical correlation function of the classical output signal, 2) normalize the calculated empirical correlation function by $\sin^2 \alpha$ and  by exponential decay determined by the measurement-induced dephasing rate $\Gamma_\bot$, 3)  apply the Leggett-Garg inequality to the  corrected empirical correlation function.

\subsection{Test experiment with a classical signal}\label{tst}

\quad

 First we apply  these ideas to  design of a test experiment in which the NV center interacts not with the nuclear spin, but with an external  classical  sinusoidal magnetic field with a random phase.
The experiment consists of a series of measurements of a classical external oscillating (linearly polarized) magnetic field (see Fig. \ref{fig2}a,c)
\begin{equation}\label{KLM}
  B_{\mathrm{ac}}(t)=B_{ac}\sin(\omega_{\mathrm{ac}} t+\phi),
\end{equation}
 where $B_{\mathrm{ac}}$ is the amplitude of the magnetic field $\omega_{ac}=2\pi f_{ac}$ is the frequency and $\phi$ is a random phase, using  the NV center as a meter.
 Since the time interval between $M_i$ is not precisely controlled, it can be assumed that each run corresponds to a different realization of the random phase $\phi$ between the oscillating field and the measurement.

The analysis of the experiment with the classical random magnetic field \ref{KLM} basically follows the approach described above, with some natural modifications. In this experiment, the effective Hamiltonian is given by
\begin{equation}
\label{eq7}
H_{\mathrm{eff}}  = 2\alpha  S_z  \sin(\omega t +\phi).
\end{equation}
We can interpret this Hamiltonian as describing the interaction of the NV sensor with a random sinusoidal signal \ref{KLM},   where $\alpha=B_z N_p \tau/\pi$ determines the strength of the interaction.

Using our theoretical approach, we calculate that the output process is given by
 \begin{equation}\label{mraser}
  z(t)=\pm\sin\big(\alpha \sin(\omega t+\phi)\big)\approx \pm\alpha \sin(\omega t+\phi).
\end{equation}
where the approximate equality holds assuming small $\alpha$.
Therefore, the correlation function of the output process $z(t)$ has the following form:
\begin{equation}
\label{eq8}
\langle z(t)z(t+\tau)\rangle =\frac{\alpha^2}{2}  \cos(\omega \tau).	
\end{equation}
Recall that the autocorrelation function of the classical  random process
$
x(t)= \sin(\omega t + \phi)
$
   with a random phase $\phi$ uniformly distributed on the interval $[0,2\pi]$ is given by
\begin{equation}\label{rjhbcb}
    C(\tau) = \frac{1}{2} \cos(\omega \tau).
\end{equation}
Comparing expressions \ref{rjhbcb} and \ref{eq8}, we find that they differ only in the factor  $\alpha^2$, which we must extract from a separate experiment.

The result of measurements in our experiment is a
sequence $\{n_k\}$ of the number of photons recorded during each readout. Thus, we need to find out how the
probabilistic properties of counting statistics $\{n_k\}$ are
related to the probabilistic characteristics of the sensor
signal. We calibrate the fluorescence output $n_a$ and $n_b$ of the NV center for electron spin in the bright $m_s  = 0$ and dark $m_s  = -1$ state (see Supplementary material,  section "SSR method").

To estimate parameter  $\alpha$  a phase modulation
$$
\phi_k  =\pi/2  \sin(2\pi k/8)
$$
is additionally applied  to the final $(\pi/2)$ pulse which modulates the output signal (see Fig. \ref{fig2}a). The empirical correlation function of the output signal has the form (Fig. \ref{fig2}b), which is the result of the superposition of the external classical field and the final modulation impulse $ \pi/2.$  We fit the curve with least square method using an analytical expression:
\begin{widetext}
\begin{equation}
\label{eq9}
\min_{n_a, n_b, \alpha} \sum_{k=1}^{200} \left( \left( \langle n_i n_{i+k}\rangle_i-n_{av}^2\right)-\frac{(n_a-n_b )^2}{4} \langle S_z^i (\alpha,\Phi_s ) S_z^{i+k} (\alpha,\Phi_s )\rangle_i \right)^2.
\end{equation}
\end{widetext}
Here $n_{av}=(n_a+n_b)/2$ and $S_z^i(\alpha ,\Phi_s)$
is the amplitude  of the component $S_z$ of the density matrix of the NV spin.  As shown in the SM (section "SSR method"), it  takes the form
\begin{equation}
S_z^i(\alpha ,\Phi_s) = \sin\left(\frac{\pi}{2}\sin\left(\frac{2 \pi k}{8}\right) + \alpha\cos\left(\frac{k \Phi_s \pi}{4}\right)\right),
\end{equation}
where  $\Phi$ (see SU, Eq. 5) is the phase acquired under the influence of the DD sequence.
As a result, we extract both the fluorescent responses of the NV center in $m_s  = 0$ ($n_a$), $m_s  = -1$ ($n_b$) and the local strength of the RF field (for details see SM section VI).
Finally, performing an experimental series without phase modulation, 
we reconstruct the normalized correlation function of the classical signal (see Fig. \ref{fig2}d), using
\begin{equation}
\langle z_iz_{i+k}\rangle_{emp}  = 4\frac{\langle n_in_{i+k}\rangle-n_{av}^2}{(n_a-n_b )^2}
\end{equation}
and Eq. \ref{eq8}.
Note that the restored correlation function of the classical signal is equal to the analytically calculated function.

\subsection{Calibration of electronic and nuclear spin parameters}\label{RcI}

To calibrate the fluorescence output $n_a$ and $n_b$ of the NV center for electron spin in the bright $m_s  = 0$ and dark $m_s  = -1$ state, we apply a phase modulation to the final $(\pi/2)$ pulse (see Fig. \ref{fig3}a). In a series averaged output, we obtain an oscillating signal, which is fitted with
\begin{equation}\label{2.17}
n(k)=\frac{1}{2}(n_a+n_b)+\frac{1}{2}(n_a-n_b)\sin^2\left(\phi_k/2 + \phi_0\right),
\end{equation}
where $\phi_k$ are the modulation angles in series $0,30,60,90...360$, $n_a$ and $n_b$ are bright and dark photon count rates (see Fig. \ref{fig3}b). The $\phi_0 \ll 1$ is the phase offset due to imperfections of the pulses due to detuning.  Measurement scheme Fig. \ref{fig3}\textbf{a} operates at $\phi_k = 90$. Each angle series consists of 50 sequential measurement, and $\phi_k = 90$ is measured 500 times (see Fig. \ref{fig3}a, number of measurements in  a inner brackets).

Then we calculate the empirical correlation function for the registered photons using the formula
$$
C_n (N)=\langle n_i n_{i+N}\rangle =\frac{1}{k-N} \sum_i^k n_i n_{i+N}
$$
 and the empirical electron spin correlation
\begin{equation}
C_{S_z} (N)= \frac{4(C_n (N)-n_{av}^2)}{(n_a-n_b)^2}
\end{equation}
using the estimated parameters  $n_a$ and $n_b$ from Fig \ref{fig3}b. We evaluate the parameters $\alpha$, using the standard least squares method $(\phi=\omega t_f )$ by comparing, the empirically estimated $C_{S_z} (N)$ to the analytical one (eq. \ref{CZ}).
In this way we find an estimate of $\alpha$ with which Eq. \ref{CZ} approximates the reconstructed correlation function of the output signal in an optimal way (see Fig. \ref{fig3}c). After carefully measuring the calibration constant
 $\alpha$ we normalized the empirical correlation
 function $C_{S_z}(N)$ by $\sin^2 \alpha$
 and get an estimation of $C_{I_x} (N)$ (see Fig. \ref{fig3}d).

\subsection{Main results}

We use  the Leggett-Garg test   in the following  form (see, for example, \cite{Emery})
\begin{equation}
\label{eq10}
LG(\tau)=2C(\tau)-C(2\tau).
\end{equation}

Let us first consider the case of a classical input sinusoidal signal with a random phase. On Fig. \ref{corr_f}\textbf{a} it can be seen that even after correcting the empirical correlation function by the factor $\alpha^2$,
as expected, there is no violation of the L-G inequality.

In the case of measurement of  the nuclear spin NV2, we first consider the case when empirical correlation function is corrected only on the factor $\sin^2\alpha$ (see Fig. \ref{corr_f}\textbf{b}). In this case the inequality is violated   until the damping effect manifests itself, that is,  only at several initial points, the number of which depends on the parameter $\alpha$.
The demonstrated result corresponds to the following experimental parameters:  KDD-XY5  sequence, free precession angle  $\approx 27^\circ$,   polarized $^{14}N$ and $\alpha\approx 0.18 \pi$. In this case  decoherence is fast and we recognize only three point of disturbance and only in the first period of oscillation.

The Fig \ref{corr_f}\textbf{c} shows the case where the empirical correlation function is corrected by an exponential decay factor. It is noteworthy that with such a correction, the correlation function demonstrates a violation of inequality in the second and even third fluctuations.
In total, we conducted 5 experiments . We used a simplified version of
memory enhanced readout proposed in \cite{neumannS}, where the NV electron spin is mapped to $^{14}N$.  We performed 200 repetitive readouts per single measurement and calculated the resulting number of photons instead of using the maximum likelihood method (see SM section V).
The  Fig. \ref{corr_f}\textbf{d} shows the result of the result of averaging  over five experiments with different control sequences. 
Again, the correlation function demonstrates a reliable violation of the inequality in the third oscillation.


Finally, we discuss the effects of ionisation on the
recorded data. It was found  that the
NV centre is in the dark state during green excitation
for $\approx 30$\% of the time without an observable fluorescence
fingerprint.  Ionization of the NV center has been identified as the limiting decoherence mechanism for quantum memories used for long-range quantum communication optical networks. However, as we emphasized above, the values of the signal in the initial period, when the backaction does not yet distort the signal, are fundamentally important for our experiment. Therefore, we select control actions that are aimed not at achieving the longest oscillation duration, but at the least signal distortion in the initial period.
We simulate the process of sequential weak measurements with initially polarized target spin using Monte-Carlo method
and  conclude that the errors induced by charge state are less than statistical errors and  do not affect our final conclusions.
A more detailed ionization mechanism and the results of numerical simulations are given in the Appendix \ref{charge}.

\emph{In conclusion}, we  showed that the empirical correlation functions corrected on the basis of our theoretical model do break the LG-inequality in different measurement regimes (see also Supplementary Information). We also showed that in a test experiment with an input random classical signal, the empirical correlation function corrected in accordance with the theory does not violate the inequality.
These results  allows us to conclude that, firstly, our model and theoretical analysis describe the experiment quite well, and secondly, that the strength parameter $\alpha$ of weak measurements is estimated experimentally with high accuracy.
This means that we correctly reconstruct the correlation of the quantum Larmor precession, which, of course, violates the LG inequality according to the theory.
The above results show that although we cannot avoid the inevitable decoherence effect during the measurement process, we can account for these changes based on theoretical analysis and accurate empirical reconstruction of the experimental parameters. 

We also note that in the process of analysis, we restore the sequence of transformations of the initial state from a purely quantum state to a classical macroscopic state at the output of the system (see section \ref{Analysis}).

\qquad

\section{Theoretical analysis of experiment}\label{Analysis}

\subsection{Notation and Preliminaries}\label{Notations}

The process of repeated (weak) measurements of some quantum observable is implemented on the composite Hilbert space $\mathcal{H}=\mathcal{H}_D\otimes \mathcal{H}_Q$ by coupling  the primary quantum system $Q$, initially in the $\rho$ state, on the Hilbert space $\mathcal{H}_Q$, to a quantum measuring device $D$, initially in the state $\sigma$.
 The two systems interact during a period $\tau$ and
after interaction, the initial density matrix  is transformed into
\begin{equation}\label{KRStraS}
\rho_{QD}(\tau)=U(\tau)\rho\otimes\sigma U^*(\tau),
\end{equation}
where $U$  is a unitary operator acting on the composite system.
The unitary group $U(\mathcal{H})$ consists of complex linear operators $U$ on $\mathcal{H},$
which satisfy $UU^* = \mathds{1}$ and, accordingly, the Lie algebra $\mathfrak{u}(\mathcal{H})$ of this group
   consists of anti-Hermitian operators.

The  manifold of general quantum states is  the family of orbits of the smooth action of the group $GL(\mathcal{H})$ of invertible operators on $\mathcal{H}$  acting on the space of self-adjoint operators according
to the map \cite{Grabowski}
\begin{equation}\label{Rank}
u^*(\mathcal{H})\ni A=A^* \rightarrow gAg^*, \quad g\in GL(\mathcal{H}).
\end{equation}
In this picture,  pure states form  an orbit in the dual space $\mathfrak{u}^*(\mathcal{H})$.
This action does not preserve the
spectrum and the trace of $A$ unless $g$ is unitary, however, it preserves the positivity of $A$ and \emph{the rank of} $A$.

In the case  of a composite $2\times 2$ system (in particular, two spin system $S=1/2$ and $I=1/2$), the basis of the state space can be written as the direct product of the basis sets of the single spins
\begin{equation}\label{Basil}
\{S_e, S_x, S_y,S_z\}\otimes\{I_e, I_x,I_y, I_z\},
\end{equation}
where $S_e=I_e=\frac{1}{2}\mathds{1},$  $I_i=\frac{1}{2}\sigma_i$.
In this case the local transformations of density matrices form a six-dimensional subgroup $SU(2)\otimes SU(2)$ of the full unitary group $U(4).$
For an isolated system with dynamical symmetry group  $SU(2)$ there exists
the corresponding (real) Lie algebra $\mathfrak{su}(2)$,  spanned by the operators  $\{\ic I_x, \ic I_y, \ic I_z\}$ satisfying
the angular momentum commutation relations
\begin{equation}\label{Crel}
 [I_x,I_y]=\ic I_z
\end{equation}
and cyclic permutations, where $I_i=\frac{1}{2}\sigma_i$ and $\sigma_i$, $i=x,y,z$ are Pauli matrices.
So  the pure state  can be expressed in terms of the basis of observables and identified with a point on the surface of the Bloch ball as the Bloch vector
\begin{equation}\label{B-vector}
\begin{aligned}
  \rho=I_e+x I_x +y I_y +z I_z:=\mathbf{w}\cdot\mathbf{I},\\
    x=\trace[\rho \ \sigma_x], \ y=\trace[\rho \ \sigma_y],  \ z=\trace[\rho \ \sigma_z]\\
  \mathbf{w}=(x,y,z), \quad  \trace[\rho]=\trace[\rho^2]=1.
  \end{aligned}
\end{equation}
The mixed states correspond to the points inside Bloch ball $|\mathbf{w}|= \trace[\rho^2]<1$. Unitary operations can be interpreted as rotations of the Bloch ball while the dissipative processes as linear or affine contractions of this ball.

It is generally considered that a quantum-mechanical system which
is isolated from the external world, has a Hamiltonian evolution. If $\mathcal{H}$ is
the Hilbert space of the system, this is expressed by the existence of a self-
adjoint (Hamiltonian) operator $H$, such that the state $\rho,$ at time $t$
is computed from the state at time zero according to the law
$
\rho_t=e^{-\ic t H}\rho e^{\ic t H}.
$
A \emph{composite} system represented in terms of Lie group  can be considered as an isolated  from the  environment.

To analyze the impact of a unitary group on the state of the system we can use the Campbell-Baker-Hausdorff (BCH) formula.
The  BCH-formula
reveals  the formal purely algebraic  connection between the local structure of a Lie group $G$ and its algebra $\mathfrak{g}$.
If one makes no further simplifying assumptions, then
the BCH formula for the orbit $e^H\rho e^H$ expands to an infinite series of nested commutators.
 But, in a particular case, under the condition $$\big[H,[H,\rho]\big] =\rho,$$
  (for example,   $ \big[S_i, [S_i,S_j]\big]=S_i$ if $S_i\neq S_j,$)
the exact formula holds
 \begin{equation}\label{AprBC}
  e^{-\ic \phi H} \rho e^{\ic \phi H} =\rho\cos \phi-\ic [H,\rho] \sin \phi.
 \end{equation}
This means that if the experiment is constructed so that the interaction Hamiltonian is expressed in terms of some basis operator $S_i$, then the dynamics of the object can be calculated using  \ref{AprBC}.
 However, the formula \ref{AprBC} does not work in the case when, as a result of interactions, the system passes into an entangled state, which is necessary to obtain information about the measurement object.

  A state $\rho_{QD}$ of a composite quantum system is called \emph{entangled} if it cannot be represented as a convex combination 
  \begin{equation}
 \rho_{QD}=\sum_\alpha p_\alpha\rho^1_\alpha\otimes \rho^2_\alpha, \quad \text{with}\quad\sum_\alpha p_\alpha=1,
 \end{equation}
 where  $\rho^1_\alpha,$ $\rho^2_\alpha$ are density matrices of the two subsystems.
Recall  that a  bipartite $\mathrm{pure}$ state $\rho_{QD}$ \emph{is entangled} if and only if its reduced states are \emph{mixed states}.
(Moreover,  a  pure state of composite system is entangled if and only if it violates Bell's inequality \cite{Gisin}, however, the assumption that violation of some Bell inequality is equivalent to the concept of entanglement is incorrect \cite{Horodecki}), \cite{Werner}).

 In our case, \emph{entanglement} manifests itself in the appearance of zero commutators in the BCH-expansion  of the composite system, such as \ref{vanish} in the \ref{bhnghXY} section.

The local properties \emph{of mixed states} of two subsystems of an entangled composite system can be studied by analyzing the \emph{action} of the elements of the local subgroup $\mathrm{Loc}=SU(2)\otimes SU(2)$ of the group $U(4)$ on state $\rho$ of a composite system lying in orbit
\begin{equation}\label{Orb}
\{\rho'=U\rho U^*, \ U\in SU(2)\otimes SU(2)\}.
\end{equation}
For this purpose, the  real symmetric $6 \times 6$ Gram matrix is introduced
 \begin{equation}\label{Gram1}
\begin{aligned}
  G_{ij}:=\frac{1}{2}\trace\big(W_i W_j\big), \quad W_j:=[R_j,\rho], \quad j=1,\ldots, 6,\\
  R_k=\ic \sigma_k \otimes \mathds{1}_2, \qquad R_{k+3}=\mathds{1}_2\otimes \ic \sigma_k, \quad k=1,2,3,
  \end{aligned}
\end{equation}
where the anti-Hermitian matrices $R_i,$ $i = 1, \ldots, 6,$ form a basis of the $\mathfrak{su}(2)\oplus\mathfrak{su}(2)$ Lie algebra.
It turns out \cite{Kus}, that the \emph{rank} of the Gram matrix is a \emph{geometric invariant} (cf. the mapping \ref{Rank}) of the orthogonal transformations, which does not change along the  orbit \ref{Orb}.
Since a state $\rho_{QD}$ of the composite system is expressed in terms of Pauli matrices, the commutators $[R_j,\rho]$ can also be represented in  terms of $[\sigma_k\otimes \mathds{1},  \sigma_j\otimes \mathds{1}]$.
Therefore, if the interaction Hamiltonian $H$ applied to a composite system in a separable state $\rho_{QD}$ generates zero terms (see \ref{vanish})
$$
 [\sigma_k\otimes \mathds{1},  \sigma_j\otimes \mathds{1}]=0
$$
 in the BCH-formula, it can serve as  \emph{evidence of the transition} of a separable state to an entangled one.
 The degeneracy of the Gram matrix and, as a consequence, the appearance of zero terms violate the conditions for the applicability of formula \ref{AprBC} .

Despite this, in our simple model, we can derive a modified exact BCH formula (see Appendix \ref{BCH-A}).
If the conditions
\begin{gather}
  B= [H,\rho], \qquad [H,B] =k\rho -k\Delta\label{ecjdbz1}\\
  [H,\Delta]=0 \qquad \big[H,[H,B]+k\Delta\big]=\big[H,[H,B]\big].\label{ecjdbz2}
\end{gather}
 are met
the following modification of the BCH formula holds:
 \begin{equation}\label{DCHGu}
 \begin{aligned}
  U\rho U^*=e^{-\ic H\phi}\rho e^{\ic H\phi}=\rho\cos (\phi\sqrt{k})+\Delta \big(1-\cos(\phi\sqrt{k})\big)\\-\frac{1}{\sqrt{k}} \ic[H,\rho] \sin(\phi\sqrt{k}).
  \end{aligned}
 \end{equation}




POV measures  naturally arise in the process of repeated measurements of a quantum observable.
The POVM method  is a realization of Naimark's theorem (see e.g.  \cite{Davis}, \cite{Holevo},\cite{Kraus}), which
states roughly that POVM scheme is equivalent to projective measurements in an extended Hilbert space (the von Neuman-L\"{u}der projection measurements). The discrepancy between the original and extended Hilbert spaces is interpreted as the presence of perturbations or inaccuracies in measurements (see Appendix \ref{POVMs}).

\emph{Completely positive, trace-preserving  maps} arise  in the POVM measurement scheme, when one wishes to restrict attention to a  subsystem $U(\mathcal{H}_Q)$
of a larger system $U(\mathcal{H}_D)\otimes U(\mathcal{H}_Q)$.  The post-measurement  state $\rho_Q$ of the \emph{primary system} $Q$
 is obtained by
projecting the joint  state $\rho_{QD}$ of the entangled system $QD$ into the subspace of quantum subsystem  by  taking \emph{partial trace} with respect to the ancilla.

The basic characterization of the measurement model is given by the quantum operation  $S_\alpha$ (see \ref{Strans} in Appendix \ref{POVMs}),  which is the linear transformation of the initial state corresponding to a projection measurement given by an orthogonal projector $P_\alpha$. The post-measurement  state of the \emph{primary system}
 is obtained by taking partial trace \ref{tracing} with respect to quantum measuring device $D$.
Thus the map $S_\alpha$ must at least be both \emph{trace-preserving} and \emph{positive-preserving} in order to preserve the density matrix property.
However, the latter is not sufficient, since $S_\alpha$ must be  the result of a \emph{positivity-preserving
process} on the larger system  $U(\mathcal{H}_D)\otimes U(\mathcal{H}_Q)$ of operators,
which is essentially an informally definition of \emph{the complete positivity} of $S_\alpha$.
Every completely positive map $S_\alpha$ can be represented (non-uniquely) in the Kraus form \cite{Kraus},
\begin{equation}\label{Kraus0}
S_\alpha(\rho)=\sum_k (M^\alpha_k)^*\rho M_k^\alpha,
\end{equation}
with some set of operators $M^\alpha_k$.
The probability to obtain result $\alpha$ in the measurement is given by (see \ref{NewmMes} for details)
\begin{equation}\label{BackProb}
\mathds{P}_{\rho\otimes\sigma}(\alpha)=\trace[\rho \sum_k(M_k^\alpha)^*M_k^\alpha]:=\trace[\rho F_\alpha],
\end{equation}
where (see \ref{Kras-Effec})
\begin{equation}\label{R-observ}
 F_\alpha= \sum_k(M_k^\alpha)^*M_k^\alpha=\trace_D  [U^*P_\alpha U\sigma].
\end{equation}
Therefore, we may identify a set of Kraus operators $\{M_k^\alpha\}$ or, equivalently, a set of effects $\{F_\alpha\}$
with a generalized observable defined by a \emph{positive operator-valued measure}
$$
R(E)=\sum_{\alpha\in E} F_\alpha, \quad E\subset \mathbb{Z}.
$$

The equation  \ref{R-observ} demonstrates that a physical quantity $F_\alpha$ of a physical system  is actually identified by the real experimental equipment used to measure it.
\emph{Thus, quantum observables, defined by $\{F_\alpha\}$ and measured relative to a
reference frame (ancillas), can be considered as relative  attributes. }

 The  representation \ref{BackProb}, which is the result of a inverse mapping of the detector output to the target system,  allows us to introduce an analogue of the classical \emph{relative entropy} of Kullback and Leibler as a measure of the discrepancy of information that occurs during the measurement process
(see definition \ref{InfEntr} and discussion in Appendix \ref{POVMs},  and application in \ref{theory}).

A sequence of (weak) POVM measurements given by a completely positive stochastic maps $S(\rho)$ generates a set of orbits inside the Bloch ball. Each measurement reduces the parameters of orbits
and sequence of measurements produce an inward-spiralling precession, which sequentially traverses the orbits.

The the stratification of the Bloch ball by the orbits can be considered as a \emph{natural quantization} generated by POVM measurements.

 \subsection{ Interaction  in   the $x$-$y$ plane }
\label{bhnghXY}

   The simple basic idea of the experiment is to study the precession of a  spin-$\frac{1}{2}$ particle  under the action of the Hamiltonian $H=-\omega I_z$  with
 some  angular precession frequency $\omega$.  The nuclear spin undergoes a free precession around
the $z$ axis with an angular velocity given by the Larmor frequency $\omega$. We expect the Leggett-Garg inequality, which limits the strength of temporal correlations in the classical structure, to be violated by the quantum mechanical unitary dynamics of the Larmor precession.
  However, in a real experiment, a challenging work is to implement and analyze the measurement procedure, which is carried out using another quantum object as a meter.
In this section, we will pass through this procedure step by step.

 \textbf{Initial condition}.
  We start the analysis of the experiment with the initial condition of the composite system corresponding to the idealized polarized state of the target nucleus spin
   \begin{equation}\label{AinitiaC}
  \rho (0)=\rho^S(0)\otimes\rho^I(0)=(S_e+S_z)\otimes(I_e + I_x).
\end{equation}
 In our experiment, we investigated both polarized and unpolarized initial conditions of the nuclear spin.
  As shown in the \ref{initST} section, an experiment with an unpolarized initial state results in partial polarization during measurements. Note also that pre-polarization procedures (see \cite{Taminiau}, \cite{Cujia}) are never 100\% efficient. How to take into account the influence of incomplete polarization on the final result is discussed in Section \ref{initST} below. We consider an experiment with a polarized initial condition as a basic idealized measurement model.

\textbf{ Step 1, sub-step 1}.
   We consider the transformation of a composite system, which is in the initial state \ref{AinitiaC}, under the influence of the Hamiltonian $H = -\omega I_z$ and apply the BCH formula to calculate the result of this transformation
   \begin{equation}\label{cmdfw}
    \rho_1(1)=(S_e+S_x)\otimes \big(I_e+I_x \cos(\omega t_f)+I_y \sin(\omega t_f)\big).
  \end{equation}
  Information about the quantum Larmor precession itself arises as amplitudes $x=\cos(\omega t_f)$ and $y=\sin(\omega t_f)$ of the observables $I_x$ and $I_y$. Since entanglement does not occur under this action, the state of the electron spin does not change.

  \textbf{Step 1, sub-step 2}.
 Next, we study the interaction
between the NV sensor and the target nuclear spin in the $x$-$y$ plane,
which under the action of the KDD sequence  is determined by the effective Hamiltonian \cite{Ma}, \cite{BossDegen}, \cite{RBL}
\begin{equation}\label{HTrth}
 H_{eff}=2\alpha S_z\otimes I_x.
\end{equation}
 The  influence  of KDD sequence  is specified by the measurement strength parameter $\alpha,$  which  depends on the perpendicular component $A_\bot$ of hyperfine field  and can be controlled by the number and  duty cycle of pulses (see section \ref{exper}). To study evolution, we apply the propagator $\exp\{-2\alpha S_z\otimes I_z\}$ to the density matrix \ref{cmdfw}.
Before applying the BCH formula, we  calculate the commutator $[H,\rho]$:
\begin{gather*}
  [H,\rho]=[2\alpha S_z\otimes I_x, (S_e+S_x)\otimes(I_e+x I_x+y I_y)]= \\
 [2\alpha S_z\otimes I_x, S_e\otimes(I_e+x I_x)] +[2\alpha S_z\otimes I_x,  S_e\otimes y I_y]\\
 [2\alpha S_z\otimes I_x, S_x\otimes I_e+x I_x] + [2\alpha S_z\otimes I_x, S_x\otimes I_e+ x I_x].
\end{gather*}
Notice, that
\begin{gather}
 [S_z\otimes I_x,  S_x\otimes I_y]=0, \label{vanish}\\
\big[S_z\otimes I_x, S_e\otimes(I_e+x I_x)\big] =[S_z,S_e]\otimes \frac{1}{2}(I_x+xI_e)=0, \label{vanish1}
 \end{gather}
 which indicates the  transition of a separable state \ref{cmdfw} into an entangled state  as a result of interaction (see Section \ref{Notations}).
Hence,
\begin{equation}\label{ghzB}
\begin{aligned}
 B:= [H,\rho]= [2\alpha S_z\otimes I_x,  S_e\otimes y I_y]\\+[2\alpha S_z\otimes I_x, S_x\otimes I_e+ x I_x]
 \\=\ic\frac{1}{2}S_y\otimes(xI_e+I_x)+\frac{\ic}{2 } S_z\otimes y I_z,
 \end{aligned}
\end{equation}
and
\begin{gather}
  [H,B]=k\Delta+k\rho, \quad \sqrt{k}=1/2, \label{jadjo}\\
  \Delta= (S_x\otimes yI_y)+S_e\otimes(I_e+x  I_x). \label{lttha}
\end{gather}
Thus,  applying the BCH formula \ref{DCHGu}
we get
 \begin{gather}
 \rho_2(1)=(\cos \alpha)\rho^S\otimes\rho^I\\ -\frac{\ic}{1/2}\sin\alpha \Big[ \frac{\ic}{2}S_y\otimes(xI_e+I_x) + \frac{\ic}{2 } S_z\otimes y I_z \Big]\\=(\sin \alpha) S_y\otimes(xI_e+I_x)\label{jo1a}\\
 +S_e \otimes(I_e+xI_x)+  (\cos\alpha) S_e \otimes yI_y\label{jo2a} \\+(\sin\alpha) S_z\otimes y I_z
  \label{jo3a}\\+(\cos\alpha) S_x\otimes(I_e+ x I_x )
+S_x\otimes y I_y,  \label{jo4a}
 \end{gather}
 which already clearly shows that entanglement has occurred.
 We note, that a similar expression   is given in \cite{Cujia}, however, it remained unclear for us how it was obtained.

This result immediately leads to \emph{a post-interaction  estimation} of the state of the nuclear spin. Following the POVM analysis rules (see  Appendix \ref{POVMs}),  we  estimate the new state $\rho^I_2(1)$ of the nuclear spin by taking the \emph{partial traces} over the NV sensor and the state of NV spin $\rho^S_2(1)$, respectively:
\begin{gather}
  \rho^I_2(1):=\trace_S [\rho_2(1)]=I_e+ I_x\cos(\omega t_f) + I_y \sin(\omega t_f)\cos \alpha, \label{postcnbv}\\
   \rho^S_2(1):=\trace_I [\rho_2(1)]=S_e+S_x\cos\alpha +S_y \cos(\omega t_f)\sin \alpha.\label{postcnbvS}
\end{gather}
 Note also that as a result of the interaction, the coordinate $x$ of the observable $I_x$ is mapped into the coordinate $\theta= \cos(\omega t_f)\sin \alpha$ of the observable $S_y$ of the NV sensor with the factor $\sin \alpha$, which specifies the magnitude of the  measurement and, as a consequence, the \emph{incompleteness} of the information obtained during the weak measurement.

\textbf{Step 1, sub-step 3}.
The rotation of the electronic spin is performed, by applying $\pi/2$ pulse along $S_x$
 \begin{gather}
 \rho_3(1)=(\sin \alpha) S_z\otimes(xI_e+I_x)\label{jo1}\\
 +S_e \otimes(I_e+xI_x)+  (\cos\alpha) S_e \otimes yI_y\label{jo2} \\-(\sin\alpha) S_y\otimes y I_z
  \label{jo3}\\+(\cos\alpha) S_x\otimes(I_e+xI_x)
+S_x\otimes yI_y.  \label{jo4}
 \end{gather}
Again, by taking the \emph{partial traces} over the NV sensor and nuclear spin we obtain
\begin{gather}
  \rho^I_3(1):=\trace_S [\rho_3(1)]=I_e+ I_x\cos(\omega t_f) + I_y \sin(\omega t_f)\cos \alpha, \label{TransI}\\
  \rho^S_3(1):=\trace_I[\rho_3(1)]=S_e+ S_x \cos \alpha +S_z \cos(\omega t_f)\sin \alpha. \label{TransS}
\end{gather}
Thus, the rotation of the electron spin leads to a mapping of the $S_y$ coordinate of the NV sensor and, consequently, the $x$-amplitude of the nuclear component $I_x$
onto the amplitude of the optically readable $S_z$ component of the NV sensor given by
$$
\zeta_1:=\trace [\sigma_z \rho^S_3(1)]=x\sin \alpha=\cos(\omega t_f)\sin \alpha.
$$

\textbf{Step 1, sub-step 4}.
The final projective measurement along the $z$ axis is the \textbf{optical readout} of the observable $S_z$
with eigenvalues and corresponding projection operators given by the formula
\begin{equation}\label{SZobs}
\begin{aligned}
 \lambda_+=\frac{1}{2}, \quad\lambda_-=-\frac{1}{2} \\
 S^\alpha=S_e+S_z,   \quad  S^\beta=S_e-S_z
  \end{aligned}
\end{equation}
These measurements are defined by the probabilities
\begin{equation}\label{Probz1}
\mathds{P}(\lambda_+)=\trace[\rho^S_3(1)S^\alpha] \quad \text{and}\quad \mathds{P}(\lambda_-)=\trace[\rho^S_3(1)S^\beta].
\end{equation}
\begin{remark}\label{problem}
Canonical projective measurement is usually viewed as an instantaneous act or spontaneous collapse of the probability amplitude. However,  the optical projective measurement is a quantum stochastic (non-unitary) process, which has been intensively studied both from a physical and mathematical point of view (see e.g. \cite{Optik}, \cite{Davis}).
We postpone the discussion of this problem to Section \ref{OPROUT} and Appendix \ref{double}.
\end{remark}

Optical readout re-polarizes the sensor back on to the initial state $\rho^S(0)=(S_e+S_z)$, while leaving the nuclear spin in the $x-y$ plane, so the next cycle of measurements we start with the next free precession  applied to the state
$$
\rho^S(0)\otimes\rho^I_3(1)= (S_e+S_z)\otimes (I_e+ I_x\cos(\omega t_f) + I_y \sin(\omega t_f)\cos \alpha).
$$
By repeating the above measurement process, we obtain  after $N$ measurements
\begin{gather}
  \rho^I(N)=I_e+  x_N I_x + y_N I_y,  \label{roI} \\
\rho^S(N)=S_e+\xi_N S_x+\theta_N S_y +\zeta_N S_z. \label{roS}
  \end{gather}
  where coordinates $x_N$ and $y_N$ are given by the   recurrent equations
\begin{gather}
 x_N= x_{N-1}\cos (\omega t_f)-y_{N-1}\sin(\omega t_f) \label{vyjuvth1}\\
 y_N= x_{N-1}\sin (\omega t_f)\cos\alpha+y_{N-1}\cos(\omega t_f)\cos\alpha, \label{vyjuvth2}
  \end{gather}
 with    $x_0=1$ and $y_0=0$. The   recorded output amplitude $\zeta_N$ is given by
 \begin{equation}\label{Trans}
 \zeta_N=\trace [\sigma_z \rho^S(N)]=x_N\sin \alpha.
\end{equation}
In terms of Bloch vectors $\vec{\rho}=(x,y,z)$, the system of equations \ref{vyjuvth1},  \ref{vyjuvth2}  can be written as follows:
 \begin{equation}\label{vfnhbwfU}
 \begin{aligned}
\vec{\rho}_I(N)=\mathcal{R}_z^\alpha\vec{\rho}_I(N-1).
   \end{aligned}
\end{equation}
where the operator $\mathcal{R}_z^\alpha$ is given by
\begin{equation}\label{Ro2a}
  \begin{bmatrix}
  \cos(\omega t_f) & -\sin(\omega t_f)& 0\\
   \sin(\omega t_f)\cos \alpha  & \cos(\omega t_f)\cos \alpha& 0\\
   0 & 0& 1\\
   \end{bmatrix}.
\end{equation}

For a sufficiently small $\alpha,$ by the same reasoning as in  \cite{Cujia}, one can obtain from \ref{vyjuvth1}, \ref{vyjuvth2} the following approximate representation of the dissipative process:
\begin{gather}
   x_N\approx \cos (\omega Nt_f)\Big(\cos\big(\alpha/2\big)\Big)^{2(N-1)}\\\approx \cos (\omega Nt_f)\exp\Big\{-\frac{(N-1)\alpha^2}{4}\Big\}, \label{htrurc}\\
  y_N\approx \sin (\omega Nt_f)\cos \alpha\Big(\cos\big(\alpha/2\big)\Big)^{2(N-1)}\\\approx \sin (\omega Nt_f)\cos\alpha\exp\Big\{-\frac{(N-1)\alpha^2}{4}\Big\}.\label{htrurc1}
\end{gather}
   This representation corresponds to the well-known form of solutions of the classical Bloch equations for the transverse components $M_x,$ $M_y$.

   Finally we note that with
   $$
   \Delta= \rho - 4\big[H,\big[H,\rho]\big] \quad\text{and }\quad \sqrt{k}=1/2,
   $$
    the closed form \ref{DCHGu} of BCH  formula can be rewritten in the form of  the master equation
 \begin{equation}\label{Lindblad}
 \begin{aligned}
   \rho_N-\rho_{N-1}=-2\ic \Big[H,\rho_{N-1}\big]\sin\alpha \\ -4\big[H, [H,\rho_{N-1}]\Big](1-\cos\alpha),
   \end{aligned}
 \end{equation}
where  $\big[H,\rho_{N-1}\big]$ is   the group generator   and
    $$
    4\Big[H, [H,\rho_{N-1}]\Big](1-\cos\alpha)
    $$
    the \emph{damping term} ($\alpha:=\phi\sqrt{k}$).
    In the specific case of our experiment, taking into account  \ref{lttha}, we can rewrite Eq. \ref{Lindblad} as follows:
 \begin{equation}\label{LidbA}
   \begin{aligned}
\rho_N=-2\ic \big[H_{eff},\rho_{N-1}\big]\sin\alpha +\rho_{N-1}\cos\alpha +\\\Big(S_x\otimes y_{N-1}I_y+S_e\otimes(I_e+x_{N-1}  I_x)\Big)(1-\cos\alpha).
 \end{aligned}
 \end{equation}
 Note that the equation  \ref {LidbA} can be considered as an analogue of the qubit master equation  in  Ref. \cite{Gambetta} (Eqs. 3.10, 3.14)  for the model of  measurements in circuit quantum electrodynamics (QED), that is,
     as a discrete time  analogue of the  master equation  in the \emph{Born-Markov  approximation in the Lindblad form}. This is not surprising, since the master equations in Refs. \cite{Lindblad} and  \cite {Kozak}, as well as in other papers, were obtained using perturbation formulas such as Lee-Trotter's formula or Kato's perturbation formula, under some additional strong assumptions.   For example,  in Refs \cite{Davis2} and  \cite{Davis}  the Markovian master equations are given under the assumption of  the weak coupling limit.
  In \cite{Lindblad} and \cite{Kozak}
 the main postulate besides the Markov property was the \emph{complete positivity} of the generator of the quantum semigroup. This is a  strict assumption, the importance of which we discussed   in connection with the concepts of entanglement  and  POVM measurements.
 Our exact formula, of course, stems from the simplicity of the particular experimental model that allows the closed BCH formula to be used.



\subsection{Optical projective readout}\label{OPROUT}

As we mentioned in the  section \ref{tst},  the result of measurements in our experiment is a
sequence $\{n_k\}$ of the number of photons recorded during each readout period.
Thus, the goal of the theory is to obtain a formula for the probability $\mathds{P}\big(m,[0,t)\big)$ that $m$ counts are recorded in the interval $[0,t)$.  (How to relate it to the probabilities \ref{Probz1} given by the theory is explained in Supplementary material, section"SSR method", see also \cite{Jorg}, \cite{Gupta})

Such formulas were obtained in the 1960s for both classical and quantum optical fields and in the phase space representation ($z=x+\ic y$) can be formally written in a similar form (see e.g. \cite{Optik}):
\begin{equation}\label{Pdetect}
  \mathds{P}\big(m,[0,t)\big)=\frac{1}{m!}\int\varphi(z)\big(|z|^2\alpha t\big)^m e^{-|z|^2\alpha t} d^2 z,
  \end{equation}
  where $\alpha$ characterize the efficiency of the detector. In this formal connection, the weight function $\varphi(z)$
  (called the Wigner quasiprobability distribution, or Glauber P-representation) is an analogue of the classical distribution function, but, in the general case, it is not non-negative (See, e.g., the well-known example of a simple harmonic oscillator introduced by Groenevold \cite{Groene} in 1946.)
Thus, the similarity in form does not mean that the quantum theory is physically equivalent to the classical theory and is called in physics the optical equivalence theorem. In cases where the measurement model leads to a positive normalized function $\varphi(z)$, it becomes a conventional normal distribution function, and formula  \ref{Pdetect} determines a meaningful distribution for all $T$.
There are examples when the negativity of the function $\varphi$  does not give rise to any difficulties in the analysis, if the basic laws of the theory are not violated. However, the question arises when $\varphi$ is the usual probability density and, consequently, as a result of measurements, we get a classical stochastic (macroscopic) process.

It follows from the analysis of section \ref{bhnghXY} that, under the influence of a sequence of weak measurements, the state of the target spin tends to a totally mixed state \ref{initial}. Thus, the transformation of a quantum random process into a classical one cannot be explained by decoherence alone. To do this, it is necessary to take into account that the process of projective optical readout is carried out with limited accuracy.

  In our experiment the basic photo-physical mechanisms behind the optical detection of the NV spin are well developed (see e.g \cite{Jorg}, \cite{Gupta}). The spin-dependence of the fluorescence arises
through an intersystem crossing to metastable singlet states, which occurs preferentially from the $m_s =\pm 1$ excited states.
The transient fluorescence signal is typically measured by counting
photons in a brief period following optical illumination. This inevitable strategy misses part of the signal because
 the differential fluorescence remains after the time cutoff, while photons arriving near the end of the counting interval are overweight.

The effect of error in an optical (non-unitary) readout process can be described using a variant of the POVM method (see Appendix \ref{double}), similar to how the occurrence of decoherence as a result of Hamiltonian transformations was demonstrated in \ref{bhnghXY}.
The Wigner distribution $\varphi(z)$ can indeed be negative, but when integrated with a certain non-negative normalized weight function $\sigma_\alpha(x,y)$ it gives a conventional probability distribution
\begin{equation}\label{SmoothW}
   \rho_\alpha(x,y)=\int_{\mathbb{R}^2}\varphi(x-\xi,y-\eta)\sigma_\alpha(\xi,\eta)d\xi d\eta.
\end{equation}
Moreover, if this weight function $\sigma_\alpha(x,y)$ depends on a  parameter   $\alpha$, which characterizes the measurement error of the detector, then using the POVM method it is possible to show that $\rho_\alpha(x,y)$ corresponds to new \emph{commutative approximate}  observables of position and momentum (see \ref {ObservQ},\ref{ObservP} in Appendix \ref{double}).
This is of course consistent with the interpretation that in the case of successive unitary transformations of a composite system, observables $F_\alpha$ given by \ref{R-observ} are generated.

 The exact meaning of this approach and the mathematical details are explained in Appendix \ref{double}.

Thus, we may conclude that the transition from a quantum process to a classical one in the optical projective measurement can be explained as a consequence of measurement inaccuracy. This, apparently, explains the fact that the results of theoretical analysis based on macroscopic model \cite{Luka} and  the classical probabilistic scheme assuming Poisson statistics and independence of repeated observation bins  \cite{Gupta} give a good approximation in the analysis of experimental data.

\subsection{Initial state generation}
\label{initST}


The special polarization procedure (see  \cite{Taminiau}, \cite {Cujia}), do not guarantee 100 \% polarization.
In this regard, below we  analyze the process which realizes an incomplete (depending on $\alpha $) polarization of the nuclear spin during the first few measurements.

In this case,  the initial condition is believed to be the thermal equilibrium state which is the result of interactions with other spins and the environment
 and is described by the totally mixed state
\begin{equation}\label{initial}
  \rho^I_0=\frac{1}{2}|0\rangle\langle 0|+\frac{1}{2}|1\rangle\langle 1|=\frac{1}{2}I_\alpha +\frac{1}{2}I_\beta=I_e.
\end{equation}
As for the sensor spin, NV center
is easily optically pumped with suitable fidelity into the
polarized state $|0\rangle.$ Thus, the initial state of the combined
system can be naturally assumed to be
\begin{equation}\label{IniCom}
 \rho^0=\rho_S(0)\otimes\rho^I_0=(S_e+S_z)\otimes I_e.
\end{equation}
After the $\pi/2$-pulse to NV spin  about the $y$-axis this state is converted to
\begin{equation}\label{IniCom1}
 \rho^0_1=(S_e+S_x)\otimes I_e.
\end{equation}
This state is separable but not absolutely separable \cite{Adhikari}, therefore we can get an entangled state with the help of a global transformation.
To this end, we first apply the  interaction controlled by  the Hamiltonian $H_{eff}=2\alpha S_z\otimes I_x$  using  microwave manipulation of the electronic spin.

To  calculate the effect of the first interaction
we apply the BCH formula to $\rho^0_1$ given by \ref{IniCom1} with
$$ [H_{eff},\rho^0_1]=\ic  S_y\otimes\frac{1}{2}I_x,  \quad \Delta=S_e\otimes I_e.$$
Thus, we obtain
\begin{equation}\label{rot0}
\rho^0_2 =S_y\sin \alpha\otimes I_x+(S_e+S_x\cos\alpha)\otimes I_e,
\end{equation}
and after the $\pi/2$ pulse along $S_x$ we get
\begin{equation}\label{rot1}
\begin{aligned}
 \rho^0_3=S_z\sin \alpha\otimes I_x+(S_e+S_x\cos\alpha)\otimes I_e.
  \end{aligned}
\end{equation}
Considering $S_z$ as an observable with eigenvalues and projections  given by \ref {SZobs}, we can \emph{predict}, using von Neumann's canonical rule (sometimes also called Bayesian estimation), the a posteriori state of the system.
 \begin{gather}
    \rho_+(0)=\frac{(S^\alpha\otimes \mathds{1})\rho^1 (S^\alpha\otimes \mathds{1})}{\mathds{P}_0^{S_z}(\lambda_+)}=(S_e+S_z)\otimes (I_e+\sin\alpha I_x)\label{apostr}\\
  \rho_-(0)=\frac{(S^\beta\otimes \mathds{1})\rho^1 (S^\beta\otimes \mathds{1})}{\mathds{P}_0^{S_z}(\lambda_-)}=(S_e-S_z)\otimes (I_e+\sin\alpha I_x),\label{apostr1}
 \end{gather}
The probabilities of
occurrence of the eigenvalues $\lambda_+$ and $\lambda_-$ are given by
\begin{equation}\label{probAB}
   \begin{aligned}
   \mathds{P}_0^{S_z}(\lambda_+)=\trace[S^\alpha \rho^1]=\frac{1}{2},\\
   \mathds{P}_0^{S_z}(\lambda_-)=\trace[S^\beta \rho^1]=\frac{1}{2}.
   \end{aligned}
\end{equation}
Since no projective readout is performed at this stage, our knowledge of the state of the system is limited to the information that, with a probability of 1/2, the nuclear spin is in a state of incomplete polarization either in direction $|0\rangle$ or in direction $|1\rangle$.
Nevertheless, this information is adequate to the  calculating the correlation function of the output process in the POVM measurement scheme.

In what follows, we will use the notation
(compare \ref{AinitiaC})
\begin{equation}\label{rand St}
  \rho_\pm(0)=(S_e+S_z)\otimes(I_e \pm I_x\sin\alpha).
\end{equation}
 for this state. Having done the calculations of Section \ref{bhnghXY} with the replacement of $x=\cos(\omega t_f)$ by $x=\pm \sin\alpha\cos(\omega t_f)$, it is easy to obtain a modification of  expression \ref{htrurc}
  \begin{equation}\label{Dinamik}
    x_N^\pm\approx \pm\sin\alpha\cos (\omega Nt_f)\exp\Big\{-\frac{(N-1)\alpha^2}{4}\Big\},
  \end{equation}
  which we use below when calculating the correlation function. In this case
the  output process is given by
  \begin{equation}\label{Trans}
  \begin{aligned}
 \zeta_N=\trace [\sigma_z \rho^S(N)]=x_N^\pm \sin \alpha\approx \\=\pm \sin^2\alpha\cos (\omega Nt_f)\exp\Big\{-\frac{(N-1)\alpha^2}{4}\Big\}.
 \end{aligned}
\end{equation}

\subsection{ Autocorrelation of  observables $I_x$ and $S_z$ and relative entropy}
\label{rdfynqbsr}\label{theory}

The autocorrelation function of a classical random process is defined as the second moment of the joint distribution.  In quantum mechanics, the definition of a joint distribution in the classical sense is meaningless due to the impossibility of (exact) simultaneous measurements of non-commuting observables in the framework of projective von Neumann measurements.
Nevertheless, this problem  is solved in terms
of the positive operator value measures  in the way, which is a natural consequence of conventional ideas of quantum theory (see \cite{Davis} and Appendix \ref{POVMs}).

We intend to calculate the correlation of the output process $\{\zeta_n\}$ given by \ref{Trans}. However, since its probabilistic properties are completely determined by the sequence $\{ x_N^\pm\}$  \ref{Dinamik}, we will deal with the calculation of correlations for this process  associated with the observable $I_x$. It means that, depending on the realized sign of the polarization, all measurements are determined either by the sequence $\{x_N^+\}$, or by the sequence $\{x_N^-\}$ .

First of all, we note that the probabilities of realizing $\{x_N^+\}$ or $\{x_N^-\}$ are determined by the probabilities \ref{probAB}
\begin{equation}\label{aB}
   \begin{aligned}
   \mathds{P}(x_N^+)=\mathds{P}^{S_z}_0\big[\lambda_+]=\frac{1}{2},\\
     \mathds{P}(x_N^+)=\mathds{P}^{S_z}_0\big[\lambda_-]=\frac{1}{2}.
   \end{aligned}
\end{equation}

Recall that observable $I_x$
has eigenvalues $\mu_+=\frac{1}{2},$ $\mu_-=-\frac{1}{2}$
with   the eigenvectors
\begin{equation}\label{sxegv}
  |0\rangle=\frac{1}{\sqrt{2}}(1,0)^T, \quad |1\rangle=\frac{1}{\sqrt{2}}(0,1)^T
\end{equation}
and corresponding projection operators $I_x^\alpha:=I_e+I_x$,  $I^\beta_x:=I_e-I_x$,  given by
\begin{equation}\label{sxegvPO}
  I_x^\alpha=\frac{1}{2}\begin{bmatrix}
    1 & 1\\
    1 & 1
 \end{bmatrix},\qquad
  I^\beta_x=\frac{1}{2}\begin{bmatrix}
    \quad1 & -1\\
    -1&\quad 1
 \end{bmatrix}.
\end{equation}
The measurements of the observable $I_x=\mu_+ I_x^\alpha  + \mu_- I^\beta_x $ corresponding to the projectors $I_x^\alpha$ and $I^\beta_x$  are given  by the probabilities (compare \ref{Probz1})
\begin{equation}\label{CondiM}
 \begin{aligned}
 \mathds{P}_N^{I_x}\big[ \mu_+|x_N^\pm\big]=\mathds{P}_N^{I_x}\big[ \mu_+|\lambda_\pm\big]=\trace\big[I_x^\alpha\ \rho^I_\pm(N)\big ],\\
\mathds{P}_{N}^{I_x}\big[\mu_-|x_N^\pm\big]=\mathds{P}_N^{I_x}\big[ \mu_-|\lambda_\pm\big]=\trace\big[I^\beta_x\ \rho^I_\pm(N) \big],
 \end{aligned}
\end{equation}
where
$
\rho^I_\pm(N)=I_e+x_N^\pm  I_x+ y_N^\pm I_y,
$
corresponds (with obvious modification ) to the density matrix \ref{roI} obtained for the case of deterministic polarization.

 Thus,  by direct calculation we find that the probabilities  $\mathds{P}_{N}^{I_x}\big[ \mu_\pm|\lambda_\pm\big]$ are given by
\begin{equation}\label{ghjjgh1}
\begin{aligned}
\mathds{P}_N^{I_x}\big[ \mu_+|\lambda_+\big]=\trace\big[I_x^\alpha\  \rho^I_+(N)\big ]=\frac{1}{2}\big(1+ x_N\big)\\
\mathds{P}_N^{I_x}\big[ \mu_+|\lambda_-\big]=\trace\big[I_x^\alpha\  \rho^I_-(N)\big ]=\frac{1}{2}\big(1- x_N\big)\\
\mathds{P}_N^{I_x}\big[ \mu_-|\lambda_+\big]=\trace\big[I^\beta_x\  \rho^I_+(N)\big ]=\frac{1}{2}\big(1- x_N\big)\\
\mathds{P}_N^{I_x}\big[ \mu_-|\lambda_-\big]=\trace\big[I^\beta_x\  \rho^I_-(N)\big ]=\frac{1}{2}\big(1+ x_N\big),
\end{aligned}
\end{equation}
where $x_N:=\sin \alpha\cos (\omega Nt_f)\exp\Big\{-\frac{(N-1)\alpha^2}{4}\Big\}$.

 Using \ref{ghjjgh1}, we define the  probabilities $p^{I_x}\big( \mu_{\pm},\lambda_{\pm}\big)$ to obtain  the value  corresponding to the state $\rho^I(N)$ \emph{together} with the value, corresponding to the state $\rho^I(0)$  by the Bayes rule
\begin{equation}\label{jointI}
\begin{aligned}
p^{I_x}\big(\mu_{+}, \lambda_+\big)= \mathds{P}_N^{I_x}\big[ \mu_+|\lambda_+\big]\ \mathds{P}^{S_z}_0\big[\lambda_{+}\big]
\\=\frac{1}{2}\big(1+ x_N\big)\ \mathds{P}^{S_z}_0\big[\lambda_{+}\big] ,\\
p^{I_x}\big(\mu_{+}, \lambda_-\big)= \mathds{P}_{N}^{I_x}\big[ \mu_+|\lambda_-\big]\ \mathds{P}^{S_z}_0\big[\lambda_{-}\big]
\\=\frac{1}{2}\big(1- x_N\big)\ \mathds{P}^{S_z}_0\big[\lambda_{-}\big] ,\\
p^{I_x}\big(\mu_{-}, \lambda_+\big)= \mathds{P}_{N}^{I_x}\big[ \mu_-|\lambda_+\big]\ \mathds{P}^{S_z}_0\big[\lambda_{+}\big]
\\=\frac{1}{2}\big(1- x_N\big)\ \mathds{P}^{S_z}_0\big[\lambda_{+}\big] ,\\
p^{I_x}\big(\mu_{-}, \lambda_-\big)= \mathds{P}_{N}^{I_x}\big[ \mu_-| \lambda_-\big]\ \mathds{P}^{S_z}_0\big[\lambda_{-}\big]
\\=\frac{1}{2}\big(1+ x_N\big)\ \mathds{P}^{S_z}_0\big[\lambda_{-}\big].
\end{aligned}
\end{equation}
Substituting  \ref{aB}  into \ref{jointI}, we obtain a set of probabilities that determine the \emph{joint distribution} $p^{I_x}\big(\mu_{\pm}, \mu_\pm\big)$:
\begin{equation}\label{Icond1}
\begin{aligned}
 p^{I_x}\big(\mu_{+}, \lambda_+\big)=\frac{1}{4}\big(1+x_N\big),  \quad p^{I_x}\big(\mu_{+}, \lambda_-\big)=\frac{1}{4}\big(1-x_N\big) ,\\
  p^{I_x}\big(\mu_{-}, \lambda_+\big)=\frac{1}{4}\big(1-x_N\big),  \quad p^{I_x}\big(\mu_{-}, \lambda_-\big)=\frac{1}{4}\big(1+x_N\big) .
  \end{aligned}
\end{equation}

We   define  the autocorrelation
of the  process corresponding to observable $I_x$  as
\begin{equation}\label{kondI}
\begin{aligned}
C^{I_x}(0,N)=
\sum_{\pm,\pm}\mu_{\pm}\lambda_{\pm}\ p^{I_x}\big( \mu_{\pm},\lambda_{\pm}\big)
\end{aligned}
\end{equation}
Substituting now expressions \ref{Icond1} in  \ref{kondI} we obtain
\begin{equation}\label{RoppI}
\begin{aligned}
  C^{I_x}(0,N)=x_N\\=\sin \alpha\cos (\omega Nt_f)\exp\Big\{-\frac{(N-1)\alpha^2}{4}\Big\}.
  \end{aligned}
  \end{equation}

Considering the connection \ref{Trans} of the registered process $\zeta_N$ and the process $x^\pm_N$, it is now easy to obtain an expression for the correlation function of the output signal
\begin{equation}\label{Czr}
\begin{aligned}
   C^{S_z}(0,N)=x^N\sin\alpha \\=\sin^2\alpha\cos (\omega Nt_f)\exp\Big\{-\frac{(N-1)\alpha^2}{4}\Big\}.
\end{aligned}
\end{equation}

In Appendix \ref{POVMs} we introduce an analogue of the classical \emph{relative entropy} of Kullback and Leibler \ref{InfEntr} as a measure of the discrepancy of information that occurs during the measurement process. 
We will now apply this formula to our particular case. 

The information that we intended to obtain during the  measurement is the initial value $x=\cos \omega t_f$ of the amplitude of the observable $I_x$.
During the measurement process, after $N$ steps, this information was translated with inevitable distortions into the amplitude $\zeta_N$ of the observable $S_z$, determined by the expression \ref{Trans}, and then was registered as a result of optical readout. Without loss of generality, we restrict ourselves for simplicity to the case of a polarized initial condition. In this case, measurements of the observable $S_z$ instead of four conditional probabilities in \ref{ghjjgh1} generate only a pair of unconditional probabilities given by the formula (recall that $\lambda_ +$, $\lambda_-$ correspond to the projections $S^\alpha$, $ S^\beta$ of the observable $S_z$)
\begin{equation}\label{EntropeEX}
\begin{aligned}
  \mathds{P}_N^{S_z}\big[ \lambda_+\big]=\frac{1}{2}\big(1+  x_N \sin \alpha\big),\\
  \mathds{P}_N^{S_z}\big[ \lambda_-\big]=\frac{1}{2}\big(1-  x_N \sin \alpha\big).
  \end{aligned}
\end{equation}
An "ideal" measurement of the observable $I_x$ in the state $\rho^I(1)=I_e+I_x \cos(\omega t_f)+I_y \sin(\omega t_f)$ gives two probabilities
$$
\mathds{P}^{I_x}=\frac{1}{2}\big(1+  \cos(\omega t_f)\big) \quad \text{and}\quad \mathds{P}^{I_x}=\frac{1}{2}\big(1-  \cos(\omega t_f)\big).
$$
Hence, we get the following expression for the relative entropy
\begin{equation}\label{EP}
\begin{aligned}
  H(S_z|I_x)=\frac{1}{2}\big(1+  x_N \sin \alpha)\cdot\log\frac{1+  x_N \sin \alpha}{1+  \cos(\omega t_f)} \\ +\frac{1}{2}\big(1-  x_N \sin \alpha)\cdot\log\frac{1-  x_N \sin \alpha}{1-  \cos(\omega t_f)}.
  \end{aligned}
\end{equation}
The Taylor expansion of the function $\log(1+y)$ gives a very simple and intuitive interpretation of the relative entropy in terms of the weighted differences $$(x_N \sin\alpha)^k-\big(\cos(\omega t_f)\big) ^k$$ of all powers of the functions $ x_N \sin\alpha$ and $\cos(\omega t_f)$ characterizing the discrepancy between the true and measured values.

\appendix

\section{Mathematical formulations of the Wigner-Bell theorems}\label{MatBell}
We will give a proof of the theorem based on formula \ref{StrogA}, which just characterizes the difference between the classical probability, for which it is valid, and the quantum (non-commutative) probability.
\begin{theorem}[\bf The Wigner-d'Espagnat inequality]\label{Wigner}
Let $\xi, \phi, \theta$ be arbitrary random variables with values $\pm 1$ on a Kolmogorov probability space $(\Omega, \mathcal{F},\mathds{P})$.
 Then the following inequality holds:
 \begin{equation}\label{Ythfd}
 \begin{aligned}
   \mathds{P}(\xi=+1, \phi=+1)+\mathds{P}(\phi=-1, \theta=+1)
   \\ \geq \mathds{P}(\xi=+1, \theta=+1).
   \end{aligned}
 \end{equation}
 \end{theorem}
 \begin{proof}
 Consider the following sets
 \begin{gather}\label{}
  A=\{\xi=1,\phi=1, \theta=1, \theta=-1\},  \label{setA}\\
   B=\{\xi=1,\xi=-1, \phi= -1\theta=1\},  \label{setB}\\
   A\cap B=\{\xi=1,\theta=1\},  \label{setAB}\\
   A\cup B=\{\xi=1,\xi=-1,\phi=1, \phi= -1, \theta=1, \theta=-1\}.\label{setAuB}
 \end{gather}
  Hence, by  \ref{StrogA}, we have
 $$
 \mathds{P}(A)+\mathds{P}(B)=\mathds{P}(A\cap B)+\mathds{P}(A\cup B).
 $$
But since the set $A$ contains both values of $\theta$ and  the set $B$ contains both values of $\xi$, this equality  is equivalent to the following representation:
\begin{equation}\label{KHR}
\begin{aligned}
  \mathds{P}\{\xi=1, \phi=1\}+\mathds{P}\{\phi=-1, \theta=1\}\\
 = \mathds{P}\{\xi=1, \theta=1\}+ \mathds{P}(A\cup B)
  \end{aligned}
\end{equation}
and due to the non-negativity of the probability, we get
\begin{equation}\label{DEsire}
 \begin{aligned}
  \mathds{P}\{\xi=1, \phi=1\}+\mathds{P}\{\phi=-1, \theta=1\}\\
 \geq \mathds{P}\{\xi=1, \theta=1\}.
 \end{aligned}
\end{equation}
\end{proof}

\section{Modification of the Baker-Campbell-Hausdorff formula }
 \label{BCH-A}
We show that under conditions
\begin{gather}
  B= [H,A], \qquad [H,B] =kA -k\Delta\label{ecjdbz1}\\
  [H,\Delta]=0 \qquad \big[H,[H,B]+k\Delta\big]=\big[H,[H,B]\big].\label{ecjdbz2}
\end{gather}
the following modification of the BCH formula holds:
 \begin{equation}\label{DCHG}
 \begin{aligned}
  UAU^*=e^{-\ic H\phi}A e^{\ic H\phi}=A\cos (\phi\sqrt{k})+\Delta \big(1-\cos(\phi\sqrt{k})\big)\\-\frac{1}{\sqrt{k}} \ic B\sin(\phi\sqrt{k}).
  \end{aligned}
 \end{equation}
By direct calculation
\begin{gather*}
 UAU^*=e^{-\ic H\phi}A e^{\ic H\phi}=A-(\ic \phi)[H,A]\\+\frac{(\ic \phi)^2}{2!}\big[H,[H,A]\big]+\frac{(\ic \phi)^3}{3!}\Big[H,\big[H, [H,A]\big]\Big]+\ldots\\
= \Big(A+ \frac{(\ic\phi)^2}{2!}[H,B]+ \frac{(\ic\phi)^4}{4!}\Big[H,\big[H,[H,B]\big]\Big]+\ldots\Big)\\
-\Big(\ic\phi[H,A]+\frac{(\ic\phi)^3}{3!}\Big[H,\big[H,B\big]\Big]+\ldots\Big)\\
= \Big(A+ \frac{(\ic\phi)^2}{2!}\big(kA-k\Delta\big)+ \frac{(\ic\phi)^4}{4!}\Big[H,\big[H,k\big(A-k\Delta\big)\big]\Big]\ldots\Big)\\
-\Big(\ic\phi[H,A]+\frac{(\ic\phi)^3}{3!} \big[H,(kA -k\Delta)\big]+\ldots\Big)
 \end{gather*}
 Since $[H,\Delta]=0$ we write
 \begin{gather*}
UAU^*=\Big(A+ \frac{(\ic\phi)^2}{2!}(kA-k\Delta)+ \frac{(\ic\phi)^4}{4!}\Big[H,\big[H,kA\big]\Big]+\ldots\Big)
\\-\Big(\ic\phi[H,A]+\frac{(\ic\phi)^3}{3!} \big[H,kA\big]+\ldots\Big)
\\=\Big(A+ \frac{(\ic\phi)^2}{2!}(kA-k\Delta)+ \frac{(\ic\phi)^4}{4!}k\Big [H,B\Big]+\ldots\Big)\\
-\Big(\ic\phi B+\frac{(\ic\phi)^3}{3!} kB+\ldots\Big)
 \end{gather*}
 Next we use $[H,B]=kA-k\Delta$ to get
 \begin{gather*}
=\Big(A+ \frac{(\ic\phi)^2}{2!}(kA-k\Delta)+ \frac{(\ic\phi)^4}{4!}k(kA -k\Delta)+\ldots\Big)
\\- B\Big(\ic\phi +\frac{(\ic\phi)^3}{3!} k+\ldots\Big)
\\=A\Big(1+ \frac{(\ic\phi)^2}{2!}k+ \frac{(\ic\phi)^4}{4!}k^2 +\ldots\Big)
\\+\Delta\Big( -\frac{(\ic\phi)^2}{2!}k- \frac{(\ic\phi)^4}{4!}k^2 +\ldots\Big)
- B\Big(\ic\phi +\frac{(\ic\phi)^3}{3!} k+\ldots\Big)\newline
\\=A\Big(1- \frac{\phi^2}{2!}k+ \frac{\phi^4}{4!}k^2 +\ldots\Big)
\\-\Delta\Big(1-1 -\frac{\phi^2}{2!}k + \frac{\phi^4}{4!}k^2 +\ldots\Big)
- B\Big(\ic\phi -\frac{(\ic\phi)^3}{3!} k+\ldots\Big)
\\
=A\Big(1- \frac{\phi^2}{2!}k+ \frac{\phi^4}{4!}k^2 +\ldots\Big)
\\+\Delta-\Delta\Big(1 -\frac{\phi^2}{2!}k + \frac{\phi^4}{4!}k^2 +\ldots\Big)\\
- B\Big(\ic\phi -\frac{(\ic\phi)^3}{3!} k+\ldots\Big)
\\=A\cos (\phi\sqrt{k})+\Delta \big(1-\cos(\phi\sqrt{k})\big)-\frac{1}{\sqrt{k}} \ic B\sin(\phi\sqrt{k}).
 \end{gather*}

\section{POVM measurements }\label{POVMs}

Let $\Omega$ be a set with a $\sigma$-field $ \mathcal{F}$,  $\mathcal{H}$ be a Hilbert space and $\mathcal{B}_{sa}(\mathcal{H})$ be a space of bounded self-adjoint operators in $\mathcal{H}$.
A positive operator valued (POV) measure on $\Omega$ is defined to be a map $F:\Omega\rightarrow \mathcal{B}_{sa}(\mathcal{H})$
 such that for $\Delta\in \mathcal{F}$,  $F(\Delta)\geq F(\varnothing)$,
 and if  $\{\Delta_n\}$ is a countable family of disjoint sets in $\mathcal{F}$ then
 $$
 F\Big(\bigcup_n \Delta_n\Big)=\sum_n F(\Delta_n),
 $$
 where the series converges in the weak operator topology.

POV measures naturally arise in the process of repeated (weak) measurements of some quantum observable (see section \ref{bhnghXY}), the scheme of which is described below. This process is implemented on the composite Hilbert space $\mathcal{H}=\mathcal{H}_D\otimes \mathcal{H}_Q$ by coupling  the primary quantum system $Q$, initially in the $\rho$ state, on the Hilbert space $\mathcal{H}_Q$, to a quantum measuring device $D$, initially in the state
\begin{equation}\label{KrsSig}
\sigma=\sum_k \lambda_k  |e_k\rangle \langle e_k|,
\end{equation}
where  the states  $|e_k\rangle$  form an orthonormal basis for the Hilbert space $\mathcal{H}_D$ of the  meter.
 The two systems interact during a period $\tau$ under the control of some Hamiltonian and the result of the interaction is described by the unitary operator $U$ \emph{acting on the composite system}.
After interaction the initial density matrix  is transformed into
\begin{equation}\label{KRStraS}
\rho_{QD}(\tau)=U(\tau)\rho\otimes\sigma U^*(\tau).
\end{equation}
   \emph{The final projection measurement} is  determined  by
orthogonal projectors $\{P_\alpha\}$
$$P_\alpha =\sum_j |\phi^\alpha_ j\rangle\langle\phi^\alpha_j|, \quad  \alpha\in \mathbb{Z},$$
associated with a \emph{measurable observable} $\mathcal{A}$ of the meter.
Here  the states $|\phi^\alpha_ j\rangle$ form an orthonormal basis for the Hilbert space $\mathcal{H}_D$ of the  meter, and  satisfy the completeness
relation
\begin{equation}\label{KRS-compl}
 \sum_{\alpha,j} |\phi^\alpha_j\rangle\langle\phi^\alpha_j|=\sum_\alpha P_\alpha=\mathds{1}_{H_D}
\end{equation}
(compare the observable $S_z$ and projections $S^\alpha=S_e+S_z$, \ $S^\beta=S_e-S_z$ in section \ref{bhnghXY}).
The post-measurement  state of the \emph{primary system}
 is obtained by taking partial trace  with respect to $D$:
\begin{equation}\label{tracing}
  \trace_D (P_\alpha U\rho\otimes\sigma U^*P_\alpha)=\trace_D (P_\alpha U\rho\otimes\sigma U^*)=S_\alpha(\rho).
\end{equation}
The probability to obtain result $\alpha$ in the measurement \emph{on the meter} is dictated by the standard von Neumann rules
for the  orthogonal projectors measurement:
\begin{equation}\label{NewmMesA}
  \mathds{P}_{\rho\otimes \sigma}(\alpha)= \trace [P_\alpha U\rho\otimes \sigma U^*]:=\trace [S_\alpha(\rho)].
\end{equation}
Thus, the basic characterization of the measurement model is given by the quantum operation  $S_\alpha$,  which is the linear transformation of the initial state
\begin{equation}\label{Strans}
\rho\rightarrow S_\alpha(\rho).
\end{equation}
Substituting \ref{KrsSig} and \ref{KRS-compl} in \ref{tracing} we get
\begin{equation}\label{krsfor}
\begin{aligned}
  S_\alpha(\rho)=\sum_{j,k} \lambda_k^{1/2}\langle\phi^\alpha_j|U|e_k\rangle\rho\langle e_k|U^*|\phi^\alpha_ j\rangle\lambda_k^{1/2}\\=\sum_{j,k}M^\alpha_{jk}\rho (M^{\alpha}_{ jk})^*.
  \end{aligned}
\end{equation}
The set of operators $M^\alpha_{ jk}=\sqrt{\lambda_k}\langle\phi^\alpha_j|U|e_k\rangle$
provides a \emph{Kraus decomposition}  of the operation $ S_\alpha$, which in turn defines
the set of effects  $F_\alpha$  given by
\begin{equation}\label{Kras-Effec}
\begin{aligned}
 F_\alpha:= \sum_{j,k}(M^{\alpha}_{ jk})^*M^\alpha_{jk}=\sum_{j,k}\lambda_k \langle e_k|U^*|\phi^\alpha_ j\rangle\langle\phi^\alpha_j|U|e_k\rangle\\=\trace_D  [U^*P_\alpha U\sigma].
 \end{aligned}
\end{equation}
The Kraus operators  and hence the set of effects $\{F_\alpha\}$ satisfy a completeness relation:
\begin{equation}\label{Kraus-comp}
\begin{aligned}
  \sum_{\alpha,j,k} (M^\alpha_{ jk})^* M^\alpha_{jk}= \sum_{\alpha,j,k}\lambda_k\langle e_k|U^*|\phi^\alpha_ j\rangle\langle\phi^\alpha_j|U|e_k\rangle\\ =\trace_D(U^*U\sigma)=\trace (\mathds{1}\otimes \sigma)=\mathds{1}.
  \end{aligned}
\end{equation}
The probability to obtain result $\alpha$  in the measurement on the ancilla can now be written  as
\begin{equation}\label{NewmMes}
\begin{aligned}
  \mathds{P}_{\rho\otimes \sigma}(\alpha)= \trace [S_\alpha(\rho)]= \trace \big[ \rho\sum_{j,k}(M^\alpha_{ jk})^* M^\alpha_{jk}\big]\\=\trace \big[\rho \cdot\trace_D  [U^*P_\alpha U\sigma]\big]=\trace [\rho F_\alpha].
  \end{aligned}
\end{equation}
Therefore, \emph{we may identify a set of effects $\{F_\alpha\}$ or, equivalently, a set of Kraus operators
  $\{M^\alpha_{jk}\}$ with a generalized observable  in the sense that
  the operator}
\begin{equation}\label{PoVM}
R(E)=\sum_{\alpha\in E}F_\alpha, \quad E \subset\mathbb{Z},\\
 \end{equation}
is a \emph{positive operator-valued measure}  and
$$
\mathds{P}_{\rho\otimes \sigma}(E)=\trace [\rho R(E)], \quad E \subset\mathbb{Z}.
$$
The equation \ref{NewmMes} demonstrates that
quantum observables are defined and measured relative
to a reference frame (ancillas) and therefore can be considered as relative attributes. 

\emph{Thus, the quantity
$F_\alpha$ of a physical system is actually identified by the real
experimental equipment used to measure the system.}
The relative nature of the observable, associated with \emph{effects} $\{F_\alpha\}$ , suggests the introduction of relative entropy as a measure of information transforming from (immeasurable) observable, related to the prime system, and a measurable observable $\mathcal{A}$, set by projections $P_\alpha$ related to the detector.

Let us assume that the observable $\mathcal{O}$ associated with the prime system, which is not accessible for direct measurement, is given by the projectors $\{Q_\alpha\}$ and that the unitary transformation $U$ in \ref{KRStraS} defining the measurement process uniquely connects the projectors $Q_\alpha$ and $\{P_\alpha\}$ by some relation (compare the observables $I_x$ and $S_z$ and their projectors  in section \ref{bhnghXY}).
By analogy with commutative  probability theory we define  the \emph{relative entropy} (called also information divergence and  introduced in classical probability by Kullback and Leibler \cite{Kullback}) by
\begin{equation}\label{InfEntr}
\begin{aligned}
  H(\mathcal{A}|\mathcal{O})=\sum_\alpha \mathds{P}_{\rho\otimes\sigma}(\alpha)\log\frac{\mathds{P}_{\rho\otimes\sigma}(\alpha)}{\mathds{P}_{\rho}(\alpha)}\\
 = \sum_\alpha \trace [\rho F_\alpha]\Big[\log\trace [\rho F_\alpha]-\log\trace [\rho Q_\alpha]\Big]
  \end{aligned}
\end{equation}
Recall that in noncommutative case for an isolated system  the relative entropy of a state $\omega$ with respect to another state $\varphi$ is usually defined in terms of the corresponding density operators by $\rho_\omega$ and $\rho_\varphi$
$$
H(\omega|\varphi)=\trace \big[\rho_\omega[ \log \rho_\omega- \log \rho_\varphi]\big]
$$

\section{The Weyl-Wigner transformation and approximate observables in the phase space representation}\label{double}

Let $H=L^2(\mathbb{R}),$ be  the Hilbert space, $\psi\in L^2(\mathbb{R})$  a state of a quantum system  and $\alpha\in L^2(\mathbb{R})$ be a function of norm one, whose expectation is zero. \emph{This function
can be regarded as specifying certain  limitations of the measurement device.}
If we define a function
\begin{equation}\label{alf}
\alpha_{xy}=e^{\ic yq}\alpha(q-x),
\end{equation}
where  the factor $\exp\{\ic yq\}$ simply maps the measurement uncertainty given by $\alpha$ from the Hilbert space into  the phase space, then for any density operator $\rho$   the \emph{non-negative} continuous function $\rho_\alpha (x,y)$ on phase space given by ($\langle\cdot,\cdot\rangle_L^2$ is a scalar product in $L^2$)
 \begin{equation}\label{jointdens}
  \rho_\alpha (x,y):=\frac{1}{2}\big\langle\rho \alpha_{xy}, \alpha_{xy}\big\rangle_{L^2},
\end{equation}
is  a probability density on $\mathbb{R}^2$
 $$
 \int_{\mathbb{R}^2}\rho_\alpha (x,y)dx dy=1.
 $$
  Recall that the   Weyl operator $W(u,v)$ is defined on $L^2(\mathbb{R})$ by
\begin{equation}\label{wign}
\Big(W(u,v)\psi\Big)(x)=\Big(e^{\ic uQ+\ic vP}\psi\Big)(x).
\end{equation}
 The following statement is due to Davis \cite{Davis}:

\emph{Let $\rho_\alpha(x, y)$ be the probability density of the state $\rho$ on phase space defined by \ref{jointdens} Then}
 \begin{equation}\label{JWdcon}
 \int_{\mathbb{R}^2}\rho_\alpha(x,y)e^{\ic xu+\ic yv}dxdy=\trace \big[\rho e^{\ic uQ+\ic vP}\big]\langle \alpha, W(u,v)\alpha\rangle_{L^2}.
 \end{equation}

 \quad

 The Wigner density $\varphi^W $ of a state $\rho$ is defined formally by the equation
\begin{equation}\label{Wigner}
\begin{aligned}
\int_{\mathbb{R}^2}\varphi^W(x,y)e^{\ic xu+\ic yv}dxdy=\trace \big[\rho e^{\ic uQ+\ic vP}\big].
\end{aligned}
\end{equation}
In other words, the Wigner density is  the inverse Fourier transform of the characteristic function $\trace \big[\rho e^{\ic uQ+\ic vP}\big].$
 If we define a function $\sigma_\alpha (x,y)$ by
  $$
  \int_{\mathbb{R}^2}\sigma_\alpha(x,y)e^{\ic xu-\ic yv}dxdy=\langle \alpha, W(u,v)\alpha\rangle_{L^2},
  $$
then, taking the Fourier transforms, we can express $ \rho_\alpha(x,y)$ as
 \begin{equation}\label{ghtlq}
 \rho_\alpha(x,y)=\int_{\mathbb{R}^2}\rho_W(x-\xi,y-\eta)\sigma_\alpha(\xi,\eta)d\xi d\eta.
 \end{equation}
  This shows that the probability density $\rho_\alpha(x,y)$  is the result of averaging the improper Wigner density by the function $\sigma_\alpha$, which depends on the  vector $\alpha\in L^2$ and reflects the inaccuracy of the measurements.

\quad

Recall that the conventional \emph{position observable} on the  Hilbert space  $\mathcal{H}=L_2(\mathbb{R})$ is the
projection-valued measure $q(\cdot)$ on $\mathbb{R},$ defined by
\begin{equation}\label{PVPosit}
\big(q(B)\psi\big)(x)=\chi_B(x)\psi(x), \quad \psi(x)\in L_2(\mathbb{R}),
\end{equation}
where $\chi_B$ is the  characteristic function of a Borel set $B\in \mathscr{B}(\mathbb{R})$. \emph{ The momentum observable} is \emph{projection-valued measure} $p(\cdot)$ on $\mathbb{R},$ defined by
\begin{equation}\label{PVMosit}
\big(p(B)\psi\big)(x)=-\ic \chi_B(x)\partial_x\psi(x).
\end{equation}

To formalize the random influence on measurements, a \emph{probability density function} (known or unknown) $f(x)$ on $\mathbb{R}$  is introduced and   the convolution
\begin{equation}\label{Conf-g}
  (f\ast g)(x)=\int_{-\infty}^\infty f(y)g(x-y)dy.
\end{equation}
is defined for a bounded measurable function $g$ on $\mathbb{R}$ .
  It can be shown (\cite{Davis}, Theorem 3.1) that  weakly convergent integrals \big(here $\widehat{f}$ is the Fourier transform of $f$ and $g\equiv\chi_E$\big)
\begin{gather}
 Q_f (E)=\int_\mathds{R} (f\ast\chi_E)(x)q(dx), \label{ObservQ} \\
 P_f (F)=\int_\mathds{R} (\widehat{f}\ast\chi_F)(k)p(dk),\label{ObservP}
\end{gather}
  uniquely define the so-called \emph{ approximate position and momentum observables} $Q_f$ and $P_f$, which   are POV measures on Borel sets $E,$ $F$ on $\mathds{R}$.  POV measures have the same properties as projection-valued (spectral) measures but \emph{take values in the set of positive  operators}, as the name suggests.
It is clear that the approximate observables defined by \ref{ObservQ}, \ref{ObservP} are commutative unsharp observables.


\section{Influence of state of charge} \label{charge}

First, we explain qualitative the effect of the charge state switching on the developed theoretical model.
We assume that each green laser pulse can fully switch the charge state, meaning that the ionisation rate at the power of 600 $\mu$W is around 3-5 MHz ($\approx 300^{-1} \, \mathrm{ns}^{-1}$), which results in switching times comparable to the duration of the laser pulse.
We assume that there is a probability $p^- \approx 70 \%$ of having $NV^-$ after the green pulse, independent on the history.
Second, we assume that the charge state is stable during the "dark time" of the measurement when the laser is switched off.
Next, we consider that the long pass 650 nm filter in the detection pass cuts most of the $NV^0$ fluorescence spectra, making it darker compared to negative charge state.
In the derivation of the correlation function, we use the fact that each measurement performs a measurement, hence imposes the back-action and causes the decay per measurement $ \alpha^2/4$.
The amplitude of the detected signal is given by the strength of the measurement $\sin \alpha$

The consequence of the charge state switching is thus two-fold: 1) the missed  $1-p^-$ part of the measurement due to NV0 state will not perturb the target spin state, making the decay of the correlation function smaller
$ \alpha^2 p^-/4$
and 2) The amplitude of the correlation function $\langle S_z S_z \rangle $ will tend to be smaller by $(p^-)^2$, due to the reduction of "useful" data, produced by NV-, and that $S_z$ produces 0 when NV0.

The latter effect is accounted by a contrast calibration procedure. The contrast shrinking of  $\Delta n = p^- (n_a - n_b)$ will be as well observed, and when photon correlation $\langle n_i n_{i+k}\rangle $ divided by the contrast $\Delta n$ the full amplitude of the $S_z$ correlation function will be restored under the assumptions that the $p^-$ stays the same under the both experiments, which is guaranteed by usage of same NV and same laser pulses.

The former is more delicate as it affects the decay constant of the nuclear spin precession under the measurements by factor of $p^-$, which clearly becomes visible on a longer correlation times.
Thus, the process of fitting the correlation function $S_z$ to the model function could give underestimated values of $\alpha$, which could lead to systematic error in its estimated values.

To show this issue quantitatively, we simulate the process of sequential weak measurements with initially polarized target spin using Monte Carlo method.
We compare results for the model with ideal case $p^- = 1.0 (100\% NV-)$ and realistic case of  $p^- = 0.7 (70\% NV-) $.

 Using the 1000 experimental runs of 100 weak measurements with $\alpha = 0.1 \pi$ corresponding roughly to KDD-XY3 sequence of NV2.
 We estimate the average $\langle I_x \rangle (N)$ evolution in both cases.
 We observe that the in the realistic case the decay is smaller, however due to noise its significance becomes clear only at long correlation times.

Next, we fit the numerically obtained $S_z$ data, with a model function $\sin \alpha \cos (\omega t_s N) \exp \{ - \alpha ^2 N/4\}$, and find $\alpha_{est}$. We divide the $S_z$ by $\sin \alpha_{est}$ and get $I_x$ reconstructed, and check the LG expression from them $LG(N) = 2 I_x(N) - I_x (2N)$

Analysing the reconstructed values we  conclude that the errors induced by charge state are less than statistical errors and do not affect significantly the correlation function relevant for the first two periods of oscillation, the region which is essential for its violation of the inequality.

Additionally, we note that it is possible to optimize the fitting procedure by weighting the correlation function points depending on their correlation length, in particular, we find that usage of the box car weighting (cutting the tail of the exponent) with window length equal to 1/3 of the original decay gives a much better estimate on the alpha values, due to the fact that on this data set, initial amplitude is more important than the decay. We point out that this could be a subject of future research in the field.

\quad

\acknowledgements

We are thankful to Professor Eric Lutz and Dr. Durga Dasari for fruitful discussions.
We acknowledge financial support by the German Science Foundation (the DFG) via SPP1601, FOR2724, the European Research Council (ASTERIQS, SMel, ERC grant 742610), the Max Planck Society, and the project QC4BW as well as QTBW.

\section*{Authors' contributions}
VV and OG developed the initial idea of the experiment. OG performed theoretical work, VV, JM performed the experimental work, HS, SO, JI synthesized and characterized the
sample, VV and OG performed data analysis, VV, OG, JW wrote the manuscript, all authors discussed and commented on the manuscripts.

\begin{figure*}
\begin{center}
\includegraphics[width =  \textwidth]{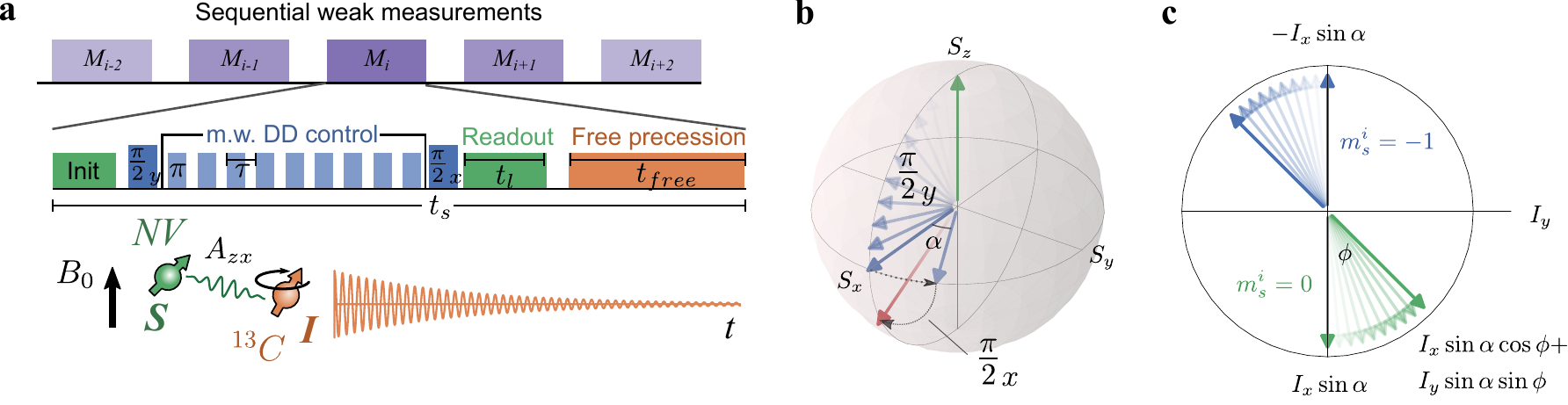}
\caption{\textbf{Scheme of the experiment}. \textbf{a)} The sequential weak measurements $M_i$ are composed of a sensor initialization part, dynamical decoupling (KDD-XY-n) and readout of the sensor state using conventional optical readout of an NV center electron spin. The KDD-XYn filter function is tuned via pulse spacing $\tau$ to the Larmor precession of the weakly coupled ($A_{zz}\approx 100 \, \mathrm{Hz}$)  $\mathrm{^{13}C}$ nuclear spin, which results in an effective interaction between nuclear and electron spin. The electron spin state is read out after each interaction, which leads to extraction of information about the nuclear spin and its back-action. \textbf{b)} Schematic evolution trajectory of an electron spin state in case the nuclear spin is in the $I_e+I_x$ state \textbf{c)} Schematic evolution of the nuclear spin during sequential measurements. Initially, the nuclear spin is in the thermally mixed state $\rho_0 = I_e$, which is then partially polarized by measurements along the $x$ axis with magnitude $\sin \alpha$. The free precession between the measurements leads to a rotation around the z axis, which is perpendicular to the figure plane. Subsequent measurements affect the measurement by disturbing both, the x and y component of the Bloch vector, conditioned on the measurement outcome (see Text).}
\label{fig1}
\end{center}
\end{figure*}

\begin{figure*}
\begin{center}
\includegraphics[width =  \textwidth]{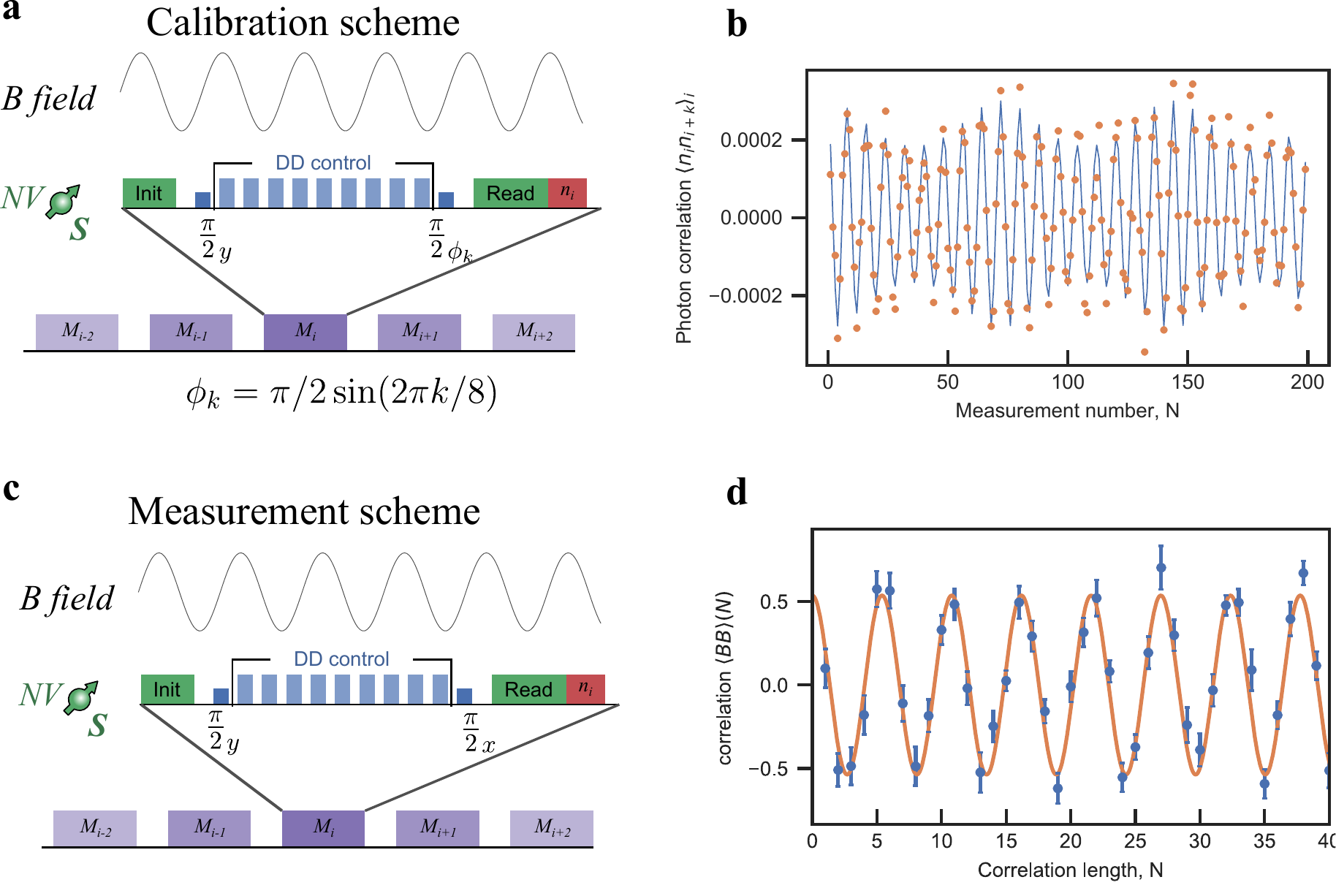}
\caption{\textbf{Externally applied classical field and calibration of fluorescence readout}. \textbf{a)} Experimental protocol for modulation assisted method for determining the fluorescence response of the NV spin readout $n_a$ and  $n_b$. Sequential measurements with optical readout of the electron spin yields phase information obtained by the interaction with the external signal. The readout $\pi/2$  pulse phase is sinusoidally modulated with an amplitude of $\pi/2$ and a period of 8 measurement cycles. \textbf{b)} The empirically calculated correlation of the centred photon counts numbers trace. The beating in the correlation originates from the presence of two frequencies. The size of the beating is determined by the relative amplitude of the external signal to phase modulation. The solid curve is the best fit of the analytical model of the correlation function \ref{eq9}, which includes phase modulation and the unknown external signal. \textbf{c)} Measurement protocol for the estimation of the classical signal correlation function. \textbf{d)} Reconstructed correlation function of the classical sinusoidal signal with a stochastic phase}
\label{fig2}
\end{center}
\end{figure*}

\begin{figure*}
\begin{center}
\includegraphics[width=\textwidth]{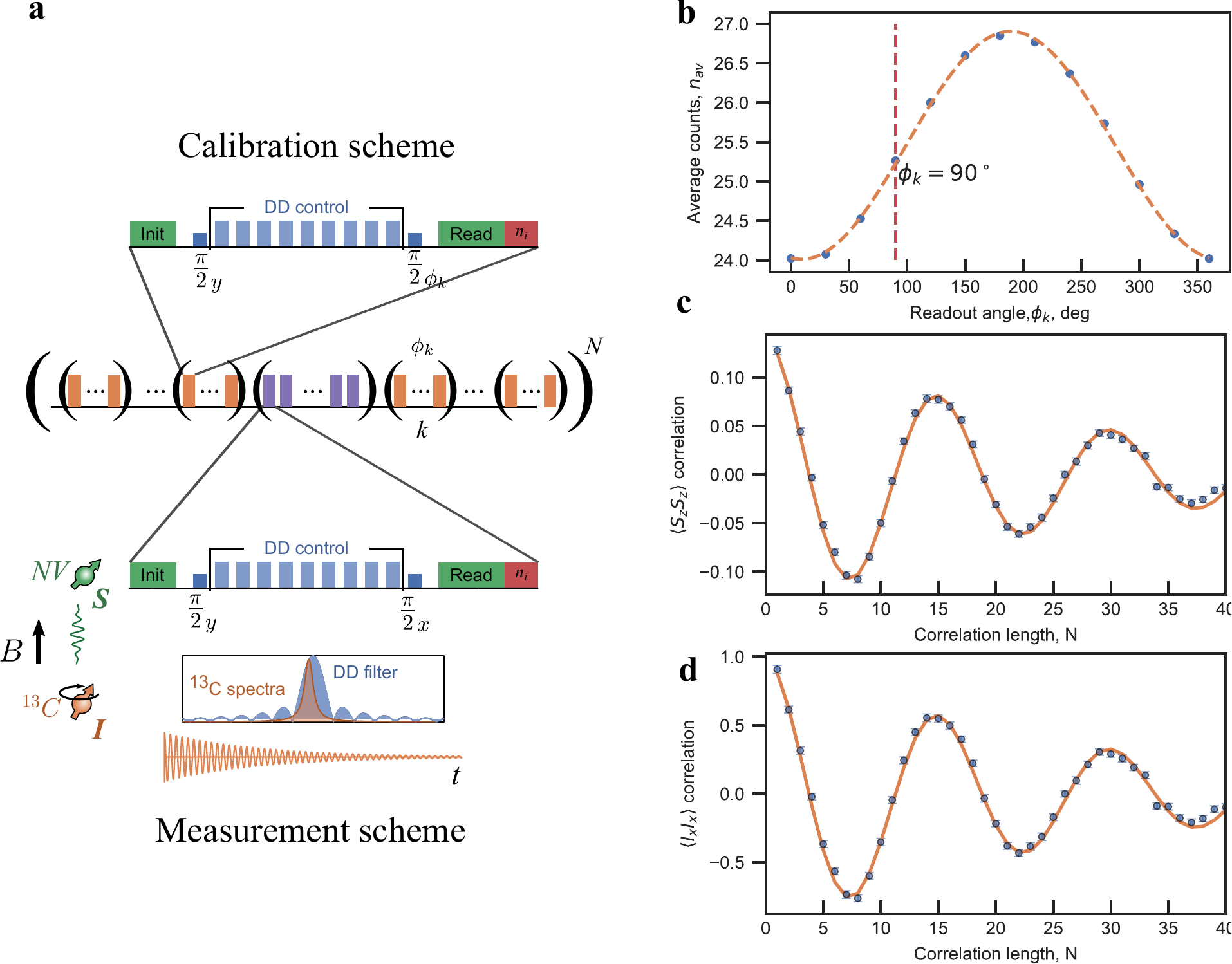}
\caption{\textbf{Reconstruction of quantum correlation function and calibration of fluorescence readout}. \newline\textbf{a)} Experimental protocol for modulation assisted method for determining fluorescence response of NV spin readout $n_a$ and  $n_b$. Sequential measurements with optical readout of electron spin results in idle measurements as DD filter function is detuned form $\mathrm{^{13}C}$ Larmor frequency. The readout $\pi/2$  pulse phase is sinusoidally modulated with amplitude $\pi/2$ and a period of 8 measurement cycles. \textbf{b)} Calibration of $n_a$  and $n_b$  for optical spin readout of the NV center electron spin. The depence of average photon counts on angle $\phi_k$ reveals the $n_a$ and $n_b$ photon counts.
The model analytical curve (see \ref{2.17}) shows that the response of the NV center depends on $n_a$ and $n_b$. Measurement scheme Fig. \ref{fig3}\textbf{a} operates at $\phi_k = 90$. \textbf{c)}. Empirically estimated correlation of sensor outputs. The solid curve is the best fit of the analytical model for correlation function which includes back-action induced decay of the initial amplitude accounted for single unknown parameter $\alpha$. \textbf{d)} Reconstructed correlation function of the quantum signal. }
\label{fig3}
\end{center}
\end{figure*}

\begin{figure*}
\begin{center}
\includegraphics[width=\textwidth]{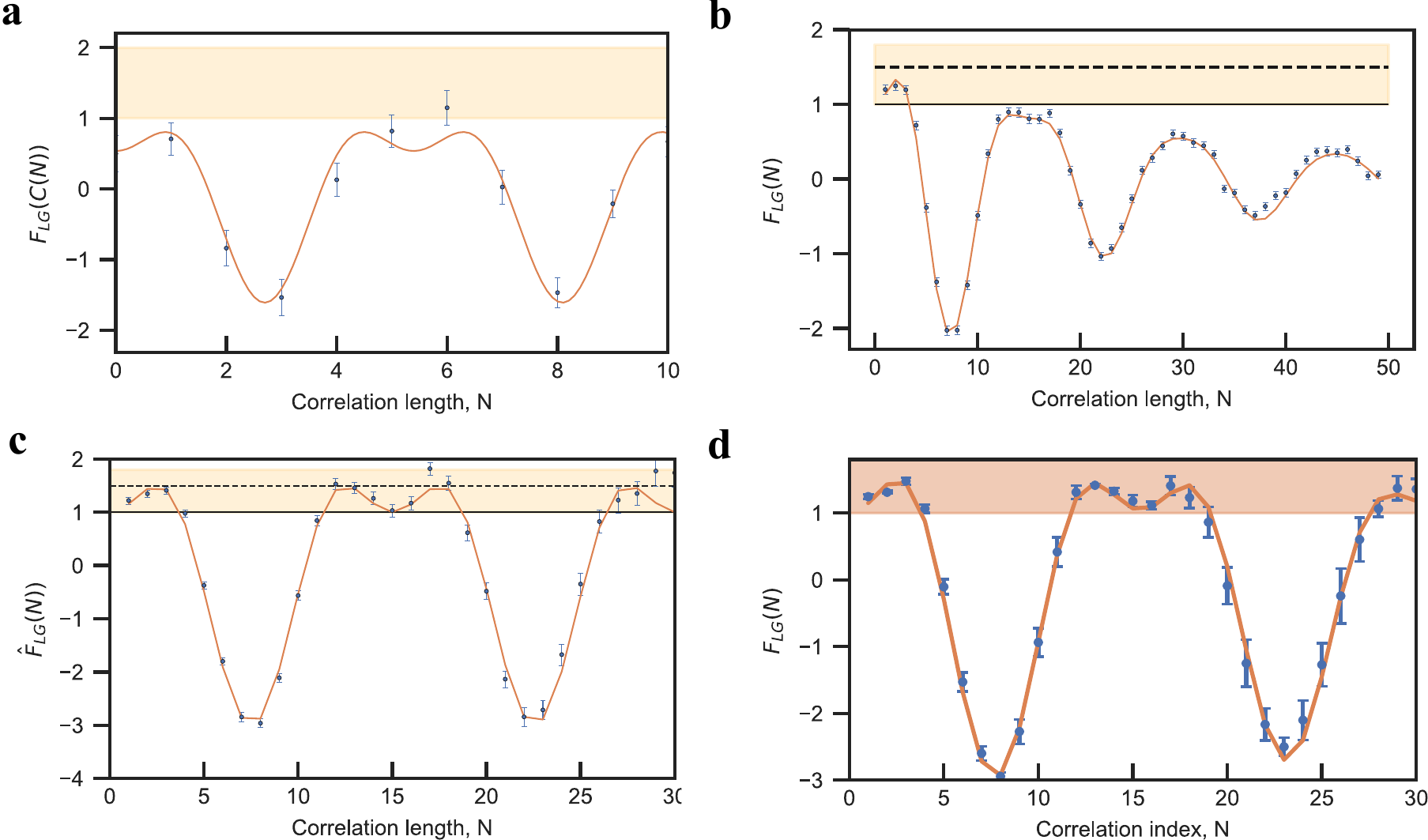}
\caption{\textbf{L-G functions for  the classical and quantum processes}. \textbf{a}) L-G function for  the  reconstructed autocorrelation of the classical r.f. signal. \textbf{b)}  NV2 measurements, L-G function for  empirical autocorrelations normalised only by $\sin^2\alpha$.  \textbf{c)} NV2 measurements, L-G function for  empirical autocorrelations normalised  by $\sin^2\alpha$ and dephasing factor \textbf{d})NV2 measurements, the result of averaging L-D functions over five experiments with different control sequences.}
\label{corr_f}
\end{center}
\end{figure*}

 \end{document}